\documentclass[aps,pra,reprint,groupedaddress]{revtex4-2}

\usepackage[utf8]{inputenc}
\usepackage[T1]{fontenc} 
\usepackage{braket}

\usepackage{amsfonts, amssymb, amsmath, amsthm}
\usepackage{mathtools}
\usepackage{dsfont}
\usepackage[american]{babel}

\usepackage{tikz}
\usetikzlibrary{shapes.geometric}
\usetikzlibrary{positioning}
\usepackage{pgfplots}
\usepackage{subfig}

\usepackage{tikz}
\usetikzlibrary{shapes}
\usetikzlibrary{positioning}
\usetikzlibrary{calc}
\usetikzlibrary{matrix}
\usetikzlibrary{arrows}

\usepackage[colorlinks=true,citecolor=blue,linkcolor=blue,urlcolor=red]{hyperref}
\usepackage{xurl}
\hypersetup{breaklinks=true}

\usepackage{enumitem}

\DeclareMathOperator{\tr}{Tr}

\DeclareMathOperator{\id}{\mathds{1}}
\DeclareMathOperator{\supp}{supp}
\DeclareMathOperator{\herm}{Herm}

\newcommand*{\KBra}[2]{|#1\rangle\!\langle#2|}

\newcommand*{\mf}{\mathfrak{F}}

\newcommand*{\mI}{\mathcal{I}}

\newcommand*{\mM}{\mathcal{M}}
\newcommand*{\mN}{\mathcal{N}}
\newcommand*{\mS}{\mathcal{S}}

\newcommand*{\sep}{\mathord{:}}

\newcounter{theorems}
\newtheorem{theorem}[theorems]{Theorem}

\newtheorem{proposition}[theorems]{Proposition}
\newtheorem{corollary}[theorems]{Corollary}
\theoremstyle{definition}
\newtheorem{definition}[theorems]{Definition}
\theoremstyle{remark}
\newtheorem{remark}{Remark}

\begin{document}

\title{Choi-Defined Resource Theories}
\author{Elia Zanoni}
\email{elia.zanoni@ucalgary.ca}
\affiliation{Department of Mathematics and Statistics, University of Calgary, Calgary, AB T2N 1N4, Canada}
\affiliation{Institute for Quantum Science and Technology, University of Calgary, Calgary, AB T2N 1N4, Canada}
\author{Carlo Maria Scandolo}
\email{carlomaria.scandolo@ucalgary.ca}
\affiliation{Department of Mathematics and Statistics, University of Calgary, Calgary, AB T2N 1N4, Canada}
\affiliation{Institute for Quantum Science and Technology, University of Calgary, Calgary, AB T2N 1N4, Canada}

\date{\today}

\begin{abstract}
	Many resource theories share an interesting property: An operation is free if and only if its renormalized Choi matrix is a free state. In this article, we refer to resource theories exhibiting this property as Choi-defined resource theories. We demonstrate how and under what conditions one can construct a Choi-defined resource theory, and we prove that when such a construction is possible, the free operations are all and only the completely resource-non-generating operations. Moreover, we examine resource measures, a complete family of monotones, and conversion distances in such resource theories.
\end{abstract}

\maketitle
\section{Introduction}
The recent booming interest in quantum computing, communication, and information originates from the understanding that quantumness allows us to outperform the classical world, which means that quantum objects are resources. The new technological developments, which are now becoming commercially available, demand a rigorous and thorough understanding of quantum resources to quantify them precisely and optimize their use. This is accomplished with the framework of resource theories~\cite{CG19, CFS16, BG15, HO13a, BBPS96, BdVSW96, VPRK97, Vid00, PV07, HHHH09, PS02, BRS07, CGMP12, MS13, JWZGB00, BHORS13, HO13, BHNOW15, MO17, SdRSFO20, CS17,SOF17,  Los19, GMNSYH15, Abe06, MS16, BCP14,VFGE12, VMGE14, HC17, ADGS18, SC19, WKRSXLGS21, WKRSXLGS21a, HG18, Gou24}. Resource theories are based on the simple observation that agents in a laboratory can only perform a subset of all possible deterministic quantum operations due to the various constraints on the agents and the systems involved. For instance, when two agents, Alice and Bob, are physically separated, they may not have a reliable means to exchange quantum systems, limiting them to local operations and classical communication (LOCC)~\cite{BBPS96, BdVSW96, VPRK97}. Such allowed operations are called \emph{free}.

While free states are typically defined as those states that can be prepared with free operations, there are situations where it is beneficial to construct free operations starting from a set of free states. For example, Alice and Bob may quickly realize that they can create all and only separable states and wonder what free operations are compatible with them. There are different ways to construct free operations starting from free states~\cite{CG19}. However, the minimal condition is that one cannot generate a resource by applying a free operation on a free state in a complete sense, i.e., even if applied only to part of a bipartite state. Operations defined with such a construction are known as completely resource non-generating (CRNG) operations \cite{CG19}. CRNG operations are of particular interest because they are the largest set of free operations compatible with a given set of free states, and they can be used to derive lower or upper bounds in any resource theory with the same set of free states. For example, separable and positive partial transpose operations have been used to overcome the challenging characterization of LOCC operations~\cite{Rai98, Cir01, HN03, Per96, HHH96, HHH98, Rai99b, APE03, IP05, MW08, GS20}. These extensions of the resource theory of entanglement are known as the resource theories of separable entanglement (SEP) and non-positive partial transpose entanglement (NPT), respectively. In other settings, such as in the resource theory of magic states~\cite{VFGE12, VMGE14, HC17, ADGS18, SC19}, imaginarity~\cite{WKRSXLGS21, WKRSXLGS21a, HG18}, and non-negativity of quantum amplitudes~\cite{JS22}, the CRNG operations are even more relevant because they coincide with the free operations.

Remarkably, one notices that in each of these instances of resource theories, a channel is considered free if and only if its renormalized Choi matrix is a free state. We call resource theories that exhibit this property \emph{Choi-defined} resource theories. This is a valuable property: Every problem associated with free quantum channels can be converted into a problem involving free states.

Given the significance of Choi-defined resource theories and their computational convenience, one wonders when it is possible to construct such resource theories from a set of free states. In this article, we answer this question by providing necessary and sufficient conditions on the set of free states. Moreover, we show that when such a construction is possible, the resulting set of free operations coincides with the set of CRNG operations associated with those free states. These results provide a constructive definition of CRNG operations and generalize what has already been observed in the resource theories of magic states, imaginarity, NPT,  SEP, and non-negativity of quantum amplitudes. Next, we introduce resource measures, a complete family of monotones, and conversion distances in CDRTs, which can all be computed with conic linear programs (CLP)~\cite{Gou24, Bar02} whenever the set of free states is convex and closed, or even with semidefinite programs (SDP)~\cite{SC23}.

\section{Preliminaries}\label{sec:prelim}
A quantum resource theory~\cite{CG19, CFS16} provides a partition of all quantum channels into free and resourceful ones, with the following properties:
\begin{enumerate}
	\item the identity channel is free,
	\item the swap channel is free,
	\item discarding a system is free,
	\item sequential composition of free channels is free,
	\item parallel composition of free channels is free.
\end{enumerate}
In other words, an agent is allowed to do nothing on a system, to change the order in which they describe systems, and to discard a system. In addition, if they perform free operations in sequence (Figure~\ref{fig:seq}) or in parallel, the resulting operation is still free. As mentioned in the introduction, a state is free if it can be prepared with free operations.  Moreover, if one applies a free channel to a free state, one obtains a free state, meaning that obtaining a resource from free objects is impossible. This is sometimes called the golden rule of resource theories~\cite{CG19}. Consequently, every set of free states satisfies the following properties: It is closed under tensor product, partial tracing, and system swapping.
\begin{figure}[htp]
	\centering
	\includegraphics[width=\columnwidth]{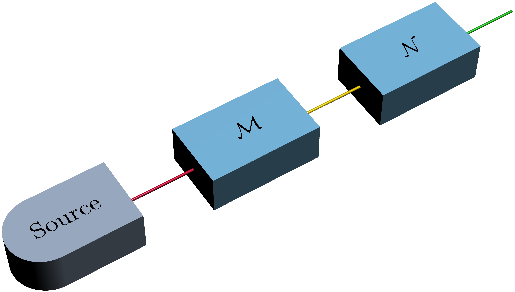}
	\caption{Sequential composition of operations. A source emits a quantum system, e.g., a photon (colored in red). A first operation $\mathcal{M}$ is performed on this quantum system and produces a new quantum system as output. A second operation $\mathcal{N}$ is performed on the new quantum system. The overall operation done on the first (red) quantum system is the sequential composition of $\mathcal{M}$ and $\mathcal{N}$, denoted with $\mathcal{N} \circ \mathcal{M}$. The sequential composition is translated into the link product when dealing with Choi matrices.} \label{fig:seq}
\end{figure}

\begin{remark}\label{rem:swap}
	One of the constraints of the resource theory of bipartite entanglement is that the two physically separated agents cannot exchange quantum systems. This may appear to conflict with the condition that the swap channel is a free operation. However, in this scenario, the systems in question are bipartite systems~\cite{CFS16}, i.e., $S_1 \coloneqq A_1B_1$ and $S_2\coloneqq A_2B_2$. The physical separation is, as usual, between systems labeled with $A$ and systems labeled with $B$. The swap channel only swaps the order of the two bipartite systems $S_1$ and $S_2$, i.e., $A_1B_1A_2B_2 \leftrightarrow A_2B_2A_1B_1$. If we organize the systems according to the physical separation, the action of the swap channel is $(A_1A_2)(B_1B_2) \to (A_2A_1)(B_2B_1)$. Therefore, the two agents are \emph{locally} exchanging the order of their own quantum systems, and there is no exchange of quantum systems between them.
\end{remark}

Free operations are mathematically described by quantum channels. Using the Choi isomorphism, a quantum channel $\mM_{A\to B}$ can be mapped to a bipartite matrix $M_{B \sep A'}$. The matrix $M_{B \sep A'}$ is known as the Choi-matrix of $\mM_{A \to B}$ and it is defined as
\begin{equation}\label{eq:def-choi-m}
	M_{B \sep A'} = \mM_{A \to B} \otimes \mI_{A'} (\Phi_{A A'}),
\end{equation}
where $\mI_A$ is the identity channel, $\Phi_{A A'} = \sum_{x, y} \KBra{x}{y}_{A} \otimes \KBra{x}{y}_{A'}$ is the \emph{unnormalized} Choi state on $A A'$, and $\Set{\Ket{x}_A}$ is an orthonormal basis for $A$. Here, we use the notation `$\sep$' to keep track of input and output. Moreover, $M_{B\sep A'}$ is positive semidefinite and $\tr_B M_{B\sep A'} = \id_{A'}$~\cite{CdAP09, BCdAP16}. From $M_{B\sep A'}$, one reconstructs the action of the quantum channel $\mM_{A \to B}$ on a state $\rho_A$ with the inverse Choi isomorphism:
\begin{equation}\label{eq:choi-inv}
	\begin{aligned}
		\mM_{A \to B}(\rho_A) & = \tr_{A'A}[(M_{B\sep A'} \otimes \rho_A)(\id_B \otimes \Phi_{A' A})] \\
		                      & = \tr_A[M_{B \sep A}(\id_B \otimes \rho_A^T)].
	\end{aligned}
\end{equation}
Note that here, one needs to know which system is the input of the original channel and which is the output.  If such information is missing, a matrix $M_{BA}$ could be associated with different linear maps, as shown in the next section.

Observe that, if we divide the Choi matrix $M_{B \sep A'}$ of the channel $\mM_{A \to B}$ by the dimension $d_A$ of $A$, we obtain a quantum state $\mu_{B\sep A'}$ such that
\begin{equation}\label{eq:rcm}
	\tr_B\mu_{B \sep A'} = \frac{1}{d_A} \id_{A'}.
\end{equation}
Consequently, a bipartite state is the \emph{renormalized Choi matrix} of a quantum channel if it satisfies the condition in Eq.~\eqref{eq:rcm}.

We point out that, Eq.~\eqref{eq:def-choi-m}, the choice of tensoring $\mathcal{M}_{A \to B}$ with the identity channel on the right is arbitrary. One can tensor on the left and define the Choi matrix of a channel $\mathcal{M}_{A \to B}$ as:
\begin{equation}
	M_{A'\sep B} = (\mI_{A'} \otimes \mM_{A \to B})(\Phi_{A'A}).
\end{equation}
In this paper, we choose Eq.~\eqref{eq:def-choi-m} as the definition of the Choi matrix of the channel $\mathcal{M}_{A \to B}$, but the treatment can be adapted to the other definition with minimal effort.

\section{Ambiguity of the Choi matrix}
In this section, we show that one needs to be careful with Choi matrices, as multiple linear maps may have the same Choi matrix. This is the key motivation for using rigorous formalism and notation for Choi matrices. Consider the matrix
\begin{equation}
	M= \frac{1}{5}\begin{pmatrix}
		1 & 2 & 0  & 0  \\
		2 & 4 & 0  & 0  \\
		0 & 0 & 4  & -2 \\
		0 & 0 & -2 & 1
	\end{pmatrix}.
\end{equation}
It is Hermitian and positive semidefinite, therefore it could be the Choi matrix of a quantum map (CP linear operator)~\cite{CdAP09}. However, if we have no information about the bipartition of such a matrix, we can find different quantum maps with $M$ as their Choi matrix. $M$ could be the Choi matrix of the quantum map that prepares the supernormalized state $M$. Another easy choice is to consider $M$ as the Choi matrix of the effect $M^T$, that acts on a $4 \times 4$ matrix $A$ as $\tr(M^T A)$.

The last choice is to consider $M$ as the Choi matrix of a quantum map from a two-dimensional system to a two-dimensional system.  In this case, there is some ambiguity as well. Let $A$ and $B$ be two-dimensional complex Hilbert spaces. We observe that we can write $M$ as a linear combination of elementary tensors in $A \otimes B$:
\begin{equation}
	M_{AB}= \frac{1}{5} \KBra{0}{0}_A \otimes \begin{pmatrix}
		1 & 2 \\
		2 & 4
	\end{pmatrix}_B + \frac{1}{5} \KBra{1}{1}_A \otimes \begin{pmatrix}
		4  & -2 \\
		-2 & 1
	\end{pmatrix}_B.
\end{equation}

In the previous section, we have presented two conventions for the definition of the Choi matrix, i.e., $M_{A \sep B} = ( \mM_{B' \to A}  \otimes \mI_{B} )(\Phi_{B'B})$ and $M_{A \sep B} = (\mI_{A} \otimes \mM_{A' \to B})(\Phi_{A A'})$ (we relabeled the systems to match with $M_{AB}$). Unsurprisingly, using different conventions, one finds different quantum maps associated with $M_{AB}$. If we consider the former, which is the one that we use in this article, we obtain the map that acts on $\KBra{0}{0}_{B'}$ as
\begin{equation}
	\mathcal{M}_{B' \to A}(\KBra{0}{0}_{B'}) = \frac{1}{5}\begin{pmatrix}
		1 & 0 \\
		0 & 4
	\end{pmatrix}_A.
\end{equation}
While with the latter, $\KBra{0}{0}_{A'}$ is mapped into
\begin{equation}
	\mathcal{M}_{A' \to B}(\KBra{0}{0}_{A'}) = \frac{1}{5}\begin{pmatrix}
		1 & 2 \\
		2 & 4
	\end{pmatrix}_B.
\end{equation}

Sometimes, if we know that a Choi matrix is the Choi matrix of a quantum channel, we can resolve this ambiguity. Indeed, the trace of the Choi matrix of a quantum channel is equal to the dimension of the input system. Therefore, since $\tr M =2$, we can rule out the first two quantum maps associated with $M$, so we must have the qubit-to-qubit map because it satisfies this condition. Another thing to consider is that if $M_{AB}$ is the Choi matrix of a quantum channel, then the marginal on the input system is the identity matrix. Consequently, we can check if either $\tr_A M_{AB}=\id_B$ or $\tr_B M_{AB} = \id_A$ to solve the ambiguity. We obtain
\begin{equation}
	\begin{aligned}
		\tr_A M & =\frac{1}{5}  \begin{pmatrix}
			                        1 & 2 \\
			                        2 & 4
		                        \end{pmatrix}_B + \frac{1}{5}  \begin{pmatrix}
			                                                       4  & -2 \\
			                                                       -2 & 1
		                                                       \end{pmatrix}_B =  \begin{pmatrix}
			                                                                          1 & 0 \\
			                                                                          0 & 1
		                                                                          \end{pmatrix}_B, \\
		\tr_B M & = \KBra{0}{0}_A + \KBra{1}{1}_A =\begin{pmatrix}
			                                           1 & 0 \\
			                                           0 & 1
		                                           \end{pmatrix}_A .
	\end{aligned}
\end{equation}
This test is inconclusive for $M$ since both marginals are the identity matrix. Therefore, even if we know that we have the Choi matrix of a quantum channel, we still have the ambiguity about the convention used to compute it.

This shows that much information is needed to reconstruct a quantum channel from its Choi matrix. It is not enough to know just the entries of a Choi matrix; one needs to know the dimensions of the input and output systems and their order in the Choi matrix. While the dimensions of the systems can always be deduced when dealing with quantum channels, the same cannot be said for the order of the systems. If one is fortunate, one can deduce it by computing partial traces. However, this is not always the case, as demonstrated for $M$.

We point out that if we consider a \emph{renormalized} Choi matrix instead of working with a Choi matrix, we lose even the information about the dimension of the input system, because all renormalized Choi matrices have trace one. For example, let $\mu_{ABC}$ be a normalized quantum state such that $\tr_A \mu_{ABC} = \frac{1}{d_Bd_C}\id_{BC}$, then also $\tr_{AB}\mu_{ABC} = \frac{1}{d_C} \id_C$, which implies that both
\begin{equation}
	\begin{aligned}
		 & \mathcal{M}_{BC \to A}(\rho_{BC}) =                                                     \\
		 & \qquad\tr_{BCB'C'}[(d_Bd_C\mu_{ABC} \otimes \rho_{B'C'})(\id_A \otimes \Phi_{BCB'C'})],
	\end{aligned}
\end{equation}
and
\begin{equation}
	\begin{aligned}
		\mathcal{M}_{C \to AB}(\rho_C) =\tr_{CC'}[(d_C\mu_{ABC} \otimes \rho_{C'})(\id_{AB} \otimes \Phi_{CC'})]
	\end{aligned}
\end{equation}
are quantum channels associated with $\mu_{ABC}$.

\section{Choi-defined resource theories}
As noted in the introduction, the resource theory of magic states, imaginarity, SEP, and NPT share an interesting property: A channel is free if and only if its renormalized Choi matrix is a free state. Now we generalize such resource theories by introducing the construction of Choi-defined operations and Choi-defined resource theories.

\begin{definition}
	The \emph{Choi-defined} (CD) operations associated with a set of free states are all and only the quantum channels such that their renormalized Choi matrix is a free state.
\end{definition}
With this definition, we are ready to define the Choi-defined resource theories.
\begin{definition}\label{def:cdrt}
	A quantum resource theory is a \emph{Choi-defined resource theory} (CDRT) if its free operations coincide with the Choi-defined operations associated with its set of free states.
\end{definition}
Not all quantum resource theories are CDRTs. For example, in the resource theory of athermality~\cite{BHORS13, HO13}, there is only one free state per system and several free quantum channels. Therefore, there is no one-to-one correspondence between free states and free channels. A more interesting example is the resource theory of entanglement. In this case, the free states are the separable states,  but the construction of CD operations produces \emph{all} separable operations, which are a larger set than LOCC~\cite{BdiVFMRSSW99, CCL12}. As a result, the resource theory of LOCC entanglement is not a CDRT.

It is a natural question to characterize when, given a set of free states, the construction of a CDRT is allowed. In every well-defined resource theory, the identity is a free operation. As a consequence, its renormalized Choi state $\frac{1}{d_{A}}\Phi_{AA'}$ is free in every CDRT.

\begin{remark}\label{rem:choi-not-ent}
	The renormalized Choi state is \emph{not} the maximally entangled state in the resource theories of entanglement, SEP entanglement, or NPT entanglement. Indeed, in these resource theories, each system is a pair of physical systems~\cite{CFS16}, i.e., $S \coloneqq AB$, and $S' \coloneqq A'B'$, and the spatial separation is between systems $AA'$ and $BB'$. Let $\Set{\Ket{a}_{A}}$ and $\Set{\Ket{b}_{B}}$ be orthogonal basis for $A$ and $B$, respectively, then
	\begin{equation}
		\frac{1}{d_S}\Phi_{S S'} = \frac{1}{d_A d_B}\sum_{a, b, \tilde{a}, \tilde{b}} \KBra{ab}{\tilde{a}\tilde{b}}_{AB} \otimes \KBra{ab}{\tilde{a}\tilde{b}}_{A'B'}.
	\end{equation}
	If the systems are reorganized according to the spatial separation, the renormalized Choi state becomes
	\begin{equation}
		\begin{aligned}
			       & \frac{1}{d_A d_B}\sum_{a, \tilde{a}} \KBra{aa}{\tilde{a}\tilde{a}}_{AA'}\otimes \sum_{b, \tilde{b}} \KBra{bb}{\tilde{b}\tilde{b}}_{BB'} \\
			\qquad & = \frac{1}{d_A}\Phi_{AA'} \otimes \frac{1}{d_B}\Phi_{BB'}.
		\end{aligned}
	\end{equation}
	This shows that the state $\frac{1}{d_{S}}\Phi_{SS'}$ is separable with respect to the spatial separation between $AA'$ and $BB'$. Therefore, it is free in the resource theories mentioned above, so it is \emph{not} the maximally entangled state on $(AA')(BB')$.
\end{remark}

When we require that resource theories be closed under the sequential composition of free channels (Figure~\ref{fig:seq}), we find a less trivial condition for the set of free states. In Ref.~\cite{CdAP09, BCdAP16}, the authors introduced an operation between bipartite matrices, called link product, defined as
\begin{equation}
	\begin{aligned}
		 & N_{C\sep B'} * M_{B \sep A'}                                                                  \\
		 & = \tr_{BB'}[(N_{C\sep B'} \otimes M_{B \sep A'} )(\id_C \otimes \Phi_{B'B} \otimes \id_{A'})] \\
		 & = \tr_B[(N_{C\sep B} \otimes \id_{A'})(\id_C \otimes M_{B \sep A}^{T_B})],
	\end{aligned}
\end{equation}
where $\cdot^{T_B}$ is the partial transpose on system $B$. This operation translates the sequential composition of channels into an operation between Choi matrices. Thanks to the link product, we can restate the golden rule of resource theories as a condition on the set of free states. That is, the condition that $\mM_{A \to B}(\rho_A)$ be free whenever $\mM_{A \to B}$ and $\rho_A$ are free becomes that $d_A \mu_{B \sep A'} * \rho_{A }$ be a free state if $\rho_A$ and $\mu_{B\sep A'}$ are free states and $\mu_{B\sep A'}$ is the renormalized Choi matrix of a quantum channel.

We now state our first main result.
\begin{theorem}\label{th:cdrt}
	It is possible to construct a CDRT associated with a set of free states if and only if for all systems $A$ and $B$
	\begin{enumerate}
		\item \label{enum:choi} the state $\frac{1}{d_{A}}\Phi_{AA'}$ is free,
		\item \label{enum:seq} if $\rho_A$ and $\mu_{B \sep A'}$ are free states, and $\mu_{B \sep A'}$ is the renormalized Choi matrix of a quantum channel, then $d_A \mu_{B \sep A'} * \rho_{A}$ is a free state.
	\end{enumerate}
\end{theorem}
\begin{proof}
	First, we prove the necessary conditions. Assume that the free states and the free operations form a well-defined CDRT.
	\begin{enumerate}
		\item The identity channel is a free operation in every resource theory. Its Choi matrix is $(\mI_A \otimes \mI_{A'})(\Phi_{A A'}) = \Phi_{A A'}$, therefore the renormalized Choi matrix $\frac{1}{d_A}\Phi_{A A'}$ is a free state.
		\item Let $\mM_{A \to B}$ be the CD operation associated with $\mu_{B \sep A'}$. Then, by Definition~\ref{def:cdrt}, $\mM_{A \to B}$ is free. Therefore, $d_A \mu_{B \sep A'} * \rho_{A} = \mM_{A \to B} (\rho_A)$ is free, because both $\mM_{A \to B}$ and $\rho_A$ are free.
	\end{enumerate}
	We now prove that conditions~\ref{enum:choi} and~\ref{enum:seq} are sufficient for a well-defined resource theory when the free operations are the Choi-defined operations. Recall that we assume that the set of free states under consideration is compatible with a minimal resource theory, i.e., closed under tensor product, partial tracing, and system swapping. Now, we demonstrate that the five conditions for a resource theory listed in Section~\ref{sec:prelim} are satisfied.
	\begin{enumerate}
		\item The identity channel is free. Indeed, since $\frac{1}{d_A} \Phi_{AA'}$ is free for all $A$, then the CD operation associated with it, the identity channel on $A$, is free.
		\item The swap channel is free. Indeed, from condition~\ref{enum:choi}, one has that $\frac{1}{d_Ad_B}\Phi_{ABA'B'}$ is a free state. Since the set of free states is closed under system swapping, one obtains that $\frac{1}{d_Ad_B}\Phi_{BA\sep A'B'}\coloneqq \frac{1}{d_{A}d_B}(\mS_{AB \to BA} \otimes \mI_{A'B'}) (\Phi_{ABA'B'})$ is free too. This state is, by definition, the renormalized Choi matrix of the swap channel, which is therefore free.
		\item Discarding a system is free. Indeed, since $\frac{1}{d_A} \Phi_{AA'}$ is free for all $A$  and the set of free states is closed under partial tracing, then $(\tr_{A} \otimes \mI_{A'}) (\frac{1}{d_A}\Phi_{AA'}) = \frac{1}{d_A}\id_A$ is free. Once again, this is by definition the Choi matrix of the discarding channel on $A$, which is therefore free.
		\item Parallel composition of free channels is free. Indeed, let $\mM_{A \to B}$ and $\mN_{C \to D}$ be free channels, and let $\mu_{BA'}$ and $\nu_{DC'}$ be their renormalized Choi matrices, respectively. The state $\tau_{BDA'C'}=(\mathcal{I}_{B \to B} \otimes \mathcal{S}_{A'D \to DA'} \otimes \mathcal{I}_{C' \to C'})(\mu_{BA'} \otimes \nu_{DC'})$, where $\mathcal{S}_{A'D \to DA'}$ is the swap channel, is free because the set of free states is closed under tensor product and system swapping. As shown in Appendix~\ref{sec:prod}, the CD operation associated with $\tau_{BDA'C'}$ is $\mM_{A \to B} \otimes \mN_{C \to D}$, which is therefore free. This proves that parallel composition of channels is free.
		\item Sequential composition of free channels is free. Let $\mM_{A \to B}$ and $\mN_{B \to C}$ be free channels, let $\mu_{B \sep A'}$ and $\nu_{C \sep B'}$ be their renormalized Choi matrices. These Choi matrices are free states. We point out that condition~\ref{enum:seq} does not immediately imply that $d_{B}\nu_{C\sep B'}*\mu_{B \sep A'}$ is free. Indeed, condition~\ref{enum:seq} only applies when the system of the second state in the link product matches the second system in the bipartition of the first state. In the case of the link product between $\nu_{C\sep B'}$ and $\mu_{B \sep A'}$, $BA'$ is not a copy of $B'$. Here, we want to construct a free state $\tilde{\nu}_{CA' \sep B'A}$ such that $d_{B}d_{A}\tilde{\nu}_{CA' \sep B'A} * \mu_{BA'}$ is the renormalized Choi matrix of $\mN_{B \to C} \circ \mM_{A \to B}$. If such construction is possible, then we can apply condition~\ref{enum:seq} of the Theorem and the definition of Choi defined resource theories to deduce that $\mN_{B \to C} \circ \mM_{A \to B}$ is free. Let $\tilde{\nu}_{CA'\sep AB'}$ be the renormalized Choi matrix of $\mN_{B \to C} \otimes \mI_{A'}$, that is $\tilde\nu_{CA' \sep B'A} = (\mathcal{I}_{C \to C} \otimes \mathcal{S}_{B'A' \to A'B'} \otimes \mathcal{I}_{A \to A})(\nu_{C\sep B'} \otimes \frac{1}{d_A}\Phi_{A' \sep A})$. Such a state satisfies the conditions above:
		      \begin{itemize}
			      \item It is free. Indeed, both $\nu_{CB'}$ and $\frac{1}{d_A}\Phi_{A'A}$ are free, and the set of free states is closed under tensor product and system swapping.
			      \item $d_{B}d_{A}\tilde{\nu}_{CA' \sep B'A} * \mu_{BA'}$ is the renormalized Choi matrix of $\mN_{B \to C} \circ \mM_{A \to B}$:
			            \begin{equation}
				            \begin{aligned}
					             & d_{B}d_{A}\tilde{\nu}_{CA' \sep B'A} * \mu_{BA'}  = (N_{B \to C} \otimes \mI_{A'})(\mu_{BA'})                   \\
					             & \qquad = \frac{1}{d_A}(N_{B \to C} \otimes \mI_{A'})\circ(\mM_{A \to B} \otimes \mI_{A'})(\Phi_{A A'})          \\
					             & \qquad = \frac{1}{d_A}\left[\left(\mN_{B \to C} \circ \mM_{A \to B} \right)\otimes \mI_{A'}\right](\Phi_{AA'}).
				            \end{aligned}
			            \end{equation}
		      \end{itemize}
	\end{enumerate}
\end{proof}

We observe that point~3 of the proof of the sufficient conditions implies the following Corollary.
\begin{corollary}\label{co:max-mix}
	In every CDRT, the maximally mixed state is free.
\end{corollary}

An interesting common property of the examples of CDRTs presented in the introduction is that their free operations coincide with the CRNG operations~\cite{Cir01, HN03, GS20, ADGS18, SC19, HG18}. As we show in the following theorem, this is a general feature of every CDRT.

\begin{theorem}\label{th:crng}
	In every CDRT, free operations coincide with CRNG operations.
\end{theorem}
\begin{proof}
	In every resource theory, free operations are CRNG (see, e.g., Ref.~\cite{CG19}). Suppose $\mM_{A \to B}$ is a CRNG operation. Since $\mM_{A \to B}$ is CRNG and $\frac{1}{d_A}\Phi_{A A'}$ is free as a consequence of Theorem~\ref{th:cdrt}, then $(\mM_{A \to B} \otimes \mI_{A'})(\frac{1}{d_A} \Phi_{A A'})$ is a free state. This state is the renormalized Choi matrix of $\mM_{A\to B}$. Therefore, it follows from Definition~\ref{def:cdrt} that $\mM_{A\to B}$ is a CD operation.
\end{proof}
Notably, this theorem provides an easy construction for CRNG operations when the set of free states satisfies the conditions of Theorem~\ref{th:cdrt}. Two examples of resource theories that are not CDRTs, and for which Theorem~\ref{th:crng} does not hold, are the resource theory of athermality~\cite{FOR15} and the resource theory of $k$-unextendibility~\cite{KDWW21}. In athermality, the CRNG operations are the Gibbs-preserving operations, but their renormalized Choi matrices are not free states. In $k$-unextendability, an attempt to construct a CDRT from $k$-extendible states was made in Ref.~\cite{PBHS13}, but the resulting operations were not resource-non-generating. Indeed, $k$-extendible states do not satisfy condition \ref{enum:choi} of Theorem~\ref{th:cdrt}.

In the remainder of this section, we present resource theories that satisfy the conditions of Theorem~\ref{th:cdrt}. We start with the resource theory of asymmetry, where, to the best of our knowledge, there is no explicit result yet in this direction~\cite{GS08, MS13}.

\subsection{Resource theories of asymmetry as CDRTs}
In the resource theory of asymmetry~\cite{GS08}, two parties have access to the same quantum systems but do not share a reference frame. The relation between reference frames is described by an element $g$ of a compact group $G$ determined by the symmetries of the problem. Free states are those that the two parties are certain to describe in the same way, that is, $G$-invariant states. In other words, as state $\rho_A$ is free if $\mathcal{U}^A_g(\rho_A) =  U^A_g \rho_A (U^A_g)^\dagger = \rho_A$ for al $g \in G$, where $U^A_g$ is the unitary representation of $g\in G$ on the Hilbert space $\mathcal{H}_A$. Note that the unitary representation on the product space $\mathcal{H}_A \otimes \mathcal{H}_B$ is  $U^{AB}_g = U^A_g \otimes U^B_g$. Similarly, a channel $\mM_{A \to B}$ is free if it is $G$-covariant, that is, $\mathcal{U}^B_g \circ \mM_{A \to B} \circ (\mathcal{U}^A_g)^{-1} = \mM_{A \to B}$ for all $g \in G$.

\begin{theorem}
	A resource theory of asymmetry is a CDRT if and only if the renormalized Choi state is free.
\end{theorem}
\begin{proof}
	Necessity is trivial because of Theorem~\ref{th:cdrt}. For sufficiency, we have only to show that the set of free states is closed under the link product whenever the Choi state is free. To this end, let $\mu_{BA'}$ and $\rho_A$ be free states such that $d_{A'}\tr_B\mu_{BA'} = \id_{A'}$. Then,
	\begin{equation}
		\begin{aligned}
			 & \mathcal{U}_g^{B}(d_A \mu_{B\sep A'}*\rho_A)  = \mathcal{U}_g^{B}(d_A \tr_{A'A}[(\mu_{BA'} \otimes \rho_A)(\id_B \otimes \Phi_{A'A})]) \\
			 & \quad  =\mathcal{U}_g^{B}(d_A \tr_{A'A}\{[(\mathcal{U}_g^{BA'A})^{-1}(\mu_{BA'} \otimes \rho_A)](\id_B \otimes \Phi_{A'A})\})          \\
			 & \quad =\mathcal{U}_g^{B}(d_A \tr_{A'A}\{[((\mathcal{U}_g^{B})^{-1}\otimes (\mathcal{U}_g^{A'A})^{-1})(\mu_{BA'} \otimes \rho_A)]       \\
			 & \qquad (\id_B \otimes \Phi_{A'A})\})                                                                                                   \\
			 & \quad = d_A \tr_{A'A}\{[\mI_B \otimes (\mathcal{U}_g^{A'A})^{-1}(\mu_{BA'} \otimes \rho_A)] (\id_B \otimes\Phi_{A'A})\}                \\
			 & \quad =d_A \tr_{A'A}\{(\mu_{BA'} \otimes \rho_A)[\id_B \otimes \mathcal{U}_g^{A'A}(\Phi_{A'A})]\}                                      \\
			 & \quad =d_A \tr_{A'A}[(\mu_{BA'} \otimes \rho_A)(\id_B \otimes \Phi_{A'A})]                                                             \\
			 & \quad = d_A\mu_{B \sep A'}*\rho_A.
		\end{aligned}
	\end{equation}
\end{proof}

In addition, the Choi state is $G$-invariant if and only if all the unitary representations are real in the Choi basis. To prove this, it is easier to work with the Choi vector $\Ket{\phi}_{AA'} = \sum_{j} \Ket{jj}_{AA'}$, rather than with the Choi matrix $\Phi_{AA'} = \KBra{\phi}{\phi}_{AA'}$. Indeed, for all matrices $M_A$ we have
\begin{equation}
	\begin{aligned}
		(M_A \otimes \id_{A'})\Ket{\phi}_{AA'} & = (M_A \otimes \id_{A'}) \sum_{j} \Ket{jj}_{AA'}                        \\
		                                       & =\sum_{j,k,l} M_{k,l} (\KBra{k}{l}_{A} \otimes \id_{A'}) \Ket{jj}_{AA'} \\
		                                       & = \sum_{j,k}M_{k,j} \Ket{kj}_{AA'}                                      \\
		                                       & = \sum_{j,k,l}M_{l,j}(\id_A \otimes \KBra{j}{l}_{A'})\Ket{kk}_{AA'}     \\
		                                       & = (\id_A \otimes M^T_{A'})\sum_k\Ket{kk}_{AA'}                          \\
		                                       & =(\id_A \otimes M^T_{A'}) \Ket{\phi}_{AA'}.
	\end{aligned}
\end{equation}
It is straightforward to see that if $U_g^A$ is real for every $g \in G$, then $(U_g \otimes U_g)\Ket{\phi}_{AA'} = (\id_A \otimes U^A_g(U_g^A)^T)\Ket{\phi}_{AA'} = \Ket{\phi}_{AA'}$ and therefore $\phi_{AA'}$ is $G$-invariant. If instead we assume that $\Ket{\phi}_{AA'}$ is $G$-invariant, then for all $g \in G$
\begin{equation}
	\begin{aligned}
		\Ket{\phi}_{AA'} & = (U^A_g \otimes U^{A'}_g)\Ket{\phi}_{AA'}                \\
		                 & = (\id_{A} \otimes U^{A'}_g(U^{A'}_g)^T)\Ket{\phi}_{AA'}.
	\end{aligned}
\end{equation}
As a consequence,
\begin{equation}
	\begin{aligned}
		\Ket{j}_{A'} & = (\Bra{j}_A \otimes \id_{A'})\Ket{\phi}_{AA'}              \\
		             & =  (\Bra{j}_A \otimes U^{A'}_g(U^{A'}_g)^T)\Ket{\phi}_{AA'} \\
		             & = U^{A'}_g(U^{A'}_g)^T\Ket{j}_{A'}
	\end{aligned}
\end{equation}
for all elements of the Choi basis. This proves that $U^{A'}_g(U^{A'}_g)^T = \id_{A'}$, and therefore the representation is real in the Choi basis.

These results are similar to those presented in the Supplementary Information of Ref.~\cite{Mar20} about completely symmetry-preserving operations. This similarity is not surprising because CD and CRNG operations coincide in CDRTs as a consequence of Theorem~\ref{th:crng}. In the context of asymmetry, an example of a CDRT is the resource theory of parity generated by the group $\mathbb{Z}_2$~\cite{GS08}.

\subsection{Example: resource theory of magic states}
The resource theory of magic states~\cite{VFGE12, VMGE14, HC17, ADGS18, SC19, WKSRXLGS24} originates from the observation that it is possible to achieve universal quantum computation by using only Clifford unitaries and some special states called magic states~\cite{BK05}. In quantum computation, systems usually consist of $n$ qubits~\cite{HC17}. In this scenario, Clifford unitaries, state preparation, and measurements in the computational basis are easy to perform and, therefore, free. In particular, Clifford unitaries are those transformations that can be implemented with circuits containing only Hadamard, CNOT, and phase-shift gates~\cite{Got98}. The free operations are called \emph{stabilizer} operations, and analogously, free states are called stabilizer states. The resources are magic states and operations.

We prove now that the set of magic states satisfies the conditions of Theorem~\ref{th:cdrt}. For the first condition, we denote the number of qubits in $A$ with $n$, and the elements of the computational basis with $\Ket{x_1 \dots x_n}$, where $ x_j = 0,1$. The normalized Choi state $\Ket{\phi}= \frac{1}{2^{n/2}}\sum_{x_1 \dots x_n} \Ket{x_1 \dots x_n} \Ket{x_{1'} \dots x_{n'}}$ is obtained by preparing the $n$ entangled states $\frac{1}{\sqrt{2}}\sum_{x_j} \Ket{x_j x_{j'}}$ and by swapping the order of the systems. The $j$-th 2-qubit entangled state can be obtained with a Hadamard and a CNOT gate acting on $\Ket{0_j0_{j'}}$, and the swap channels can be decomposed into three CNOT gates (see, e.g., Ref.~\cite{NC10}). All these operations are stabilizer operations; therefore, the renormalized Choi state $\frac{1}{d_{A}}\Phi_{AA'}$ is a stabilizer state. This proves that condition~\ref{enum:choi} of Theorem~\ref{th:cdrt} is satisfied.

We now consider the closure under the link product of the set of stabilizer states. Let $\mu_{B A'}$ and $\rho_A$ be free states such that $d_{A'} \tr_B\mu_{BA'}= \id_{A'}$, where $A$ is a $n$-qubit system and $B$ is a $m$-qubit system. This implies that there exists a protocol to create $\mu_{B A'}  \otimes {\rho_A}$ for free. By Eq.~\eqref{th:cdrt}, the state $d_A \mu_{B \sep A'} * \rho_A$ is obtained by a protocol with postselection:
\begin{enumerate}
	\item Prepare $\mu_{B A'} \otimes \rho_A$.
	\item Perform a measurement in the Bell basis for each of the pairs of qubits in $A'A$
	\item Postselect on the measurement outcome associated with $\frac{1}{d_A}\Phi_{A' A}$. After the postselection, the system $B$ is in the state $d_A\mu_{B\sep A'} * \rho_A$, as requested.
\end{enumerate}
Each of the steps of this protocol is a free operation. Indeed, the measurement in the Bell basis can be converted into a measurement in the computational basis by adding Hadamard and CNOT gates, which are stabilizer operations. Therefore the state $d_A \mu_{B \sep A'} * \rho_A$ is free, as requested by Theorem~\ref{th:cdrt}. The set of stabilizer states satisfies all the conditions of Theorem~\ref{th:cdrt}, and the CDRT constructed from it is the resource theory of magic states~\cite{SC19}.

\subsection{Example: resource theory of imaginarity}
We now focus on the resource theory of imaginarity~\cite{WKRSXLGS21, WKRSXLGS21a, HG18, WKSRXLGS24}. The aim of this resource theory is to evaluate the role of imaginarity in quantum mechanics. A state is free if its density matrix, expressed with respect to a fixed basis, is real. States with density matrices containing non-real numbers are resources. The Choi state $\frac{1}{d_A}\Phi_{AA'}$ relative to the fixed basis has only real coefficients, and therefore is free. This proves condition~\ref{enum:choi} of Theorem~\ref{th:cdrt}. Condition~\ref{enum:seq} follows trivially because it can be expressed using only multiplications and traces of real matrices, from which it is possible to obtain only real matrices. This proves that it is possible to construct a CDRT from the set of real matrices, and it coincides with the resource theory of imaginarity, where the free operations are all the CRNG operations~\cite{HG18}.

\subsection{Example: resource theory of separable entanglement}
In the theory of entanglement~\cite{BBPS96, BdVSW96, VPRK97, Vid00, PV07, HHHH09}, systems are bipartite systems, and free states are separable states. We have already shown in Remark~\ref{rem:choi-not-ent} that the Choi state is separable, and therefore free. We now prove that the separable states are closed under link product. To align with the convention used in entanglement-related resource theories, we denote the first pair of systems with $S_1 \coloneqq A_1 B_1$ and the second with $S_2 \coloneqq A_2 B_2$. Let $\mu_{S_2 S_1'}$ be a separable state, i.e., there exists density matrices $\mu^j_{A_2A'_1}$ and $\tilde\mu^j_{B_2B'_1}$ such that $\mu_{A_2A'_1B_2B'_1} = \sum_j p_j \mu^j_{A_2A'_1} \otimes \tilde\mu^j_{B_2 B'_1}$, where $p_j \geq 0$ and $\sum_j p_j =1$. Moreover, let $\mu_{S_2 S_1'}$ be such that $d_{S_1'}\tr_{S_2}\mu_{S_2 S_1'} = \id_{S'}$. Let $\rho_{S_1} = \sum_k q_k \rho^k_{A_1}  \otimes \tilde\rho^k_{B_1}$, where $q_k \geq 0$ and $\sum_k q_k =1$, be another separable state. Then
\begin{equation}
	\begin{aligned}
		\mu_{S_2 \sep S_1'} * \rho_{S_1} & = \tr_{(A_1B_1)}[\mu_{(A_2B_2) (A_1B_1)} (\id_{(A_2B_2)}\otimes \rho_{(A_1B_1)}^T )] \\
		                                 & = \sum_{j,k}p_jq_k\sigma^{j,k}_{A_2} \otimes  \tilde\sigma^{j,k}_{B_2}.
	\end{aligned}
\end{equation}
where $\sigma^{j,k}_{A_2} \coloneqq \tr_{A_1}\{\mu^j_{A_2A_1}[\id_{A_2} \otimes (\rho^k_{A_1})^T] \}$, and $\tilde\sigma^{j,k}_{B_2} \coloneqq  \tr_{B_1}\{ \tilde\mu^j_{B_2B_1}[\id_{B_2} \otimes (\tilde\rho^k_{B_1})^T]\}$.
This proves that $d_{S_1}\mu_{S_2 S_1'} * \rho_{S_1} $ is a separable state whenever $\mu_{S_2 S_1'}$ and $\rho_{S_1}$ are separable, therefore condition~\ref{enum:seq} of Theorem~\ref{th:cdrt} is satisfied. This implies that it is possible to construct a CDRT from the set of separable states, which coincides with SEP~\cite{Cir01, HN03}.

\subsection{Example: resource theory of NPT entanglement}
Looking again at the resource theory of entanglement, one notices that all separable states have positive partial transpose~\cite{Per96, HHH96}. One can thus consider the bipartite states with positive partial transpose to be free states. This set is strictly larger than the set of separable states due to bound entanglement~\cite{HHH98}, which is of great relevance in quantum information~\cite{Per96, HHH96, HHH98, Rai99b, APE03, IP05, MW08}. Since the set of PPT states contains separable states, it automatically satisfies condition~\ref{enum:choi} of Theorem~\ref{th:cdrt}. Let $\mu_{(A_2B_2) (A_1B_1)'}$ and $\rho_{(A_1B_1)}$ be free states such that $d_{(A_1B_1)'}\tr_{(A_2, B_2)}\mu_{(A_2B_2) (A_1B_1)'}= \id_{(A_1B_1)'}$. We write $(\mu_{(A_2B_2) (A_1B_1)'} * \rho_{(A_1B_1)})^{T_{B_2}}$ as
\begin{equation}
	\begin{aligned}
		 & (\mu_{(A_2B_2) \sep(A_1B_1)'} * \rho_{(A_1B_1)})^{T_{B_2}}                                                       \\ &\qquad= \tr_{A_1B_1} (\mu_{(A_2B_2) (A_1B_1)}^{T_{B_2}} ( \id_{(A_2B_2)} \otimes \rho_{(A_1B_1)}^T))            \\
		 & \qquad =  \tr_{A_1B_1}[\mu_{(A_2B_2) (A_1B_1)}^{T_{B_2B_1}}( \id_{(A_2B_2)} \otimes \rho_{(A_1B_1)}^{T_{A_1}})].
	\end{aligned}
\end{equation}
Since $\rho_{(A_1B_1)}$ has positive partial transpose, there exists and eigenbasis $\Set{\Ket{\varphi^j_{(A_1B_1)}}}$ and positive eigenvalues $\lambda_j$ such that $\rho_{(A_1B_1)}^{T_{A_1}} = \sum_j \lambda_j \KBra{\varphi^j_{(A_1B_1)}}{\varphi^j_{(A_1B_1)}}$. This implies that, for all $\Ket{\psi_{(A_2B_2)}}$,
\begin{equation}
	\begin{aligned}
		 & \Braket{\psi_{(A_2B_2)}| (\mu_{(A_2B_2)\sep (A_1B_1)'}  * \rho_{(A_1B_1)})^{T_{B_2}} | \psi_{(A_2B_2)}} =                                                   \\
		 & \sum_j \lambda_j \Braket{\psi_{(A_2B_2)} \otimes \varphi_{(A_1B_1)}^j| \mu_{(A_2B_2) (A_1B_1)}^{T_{B_2B_1}} | \psi_{(A_2B_2)} \otimes \varphi_{(A_1B_1)}^j} \\
		 & \qquad \geq 0,
	\end{aligned}
\end{equation}
where the last inequality follows from the fact that $ \mu_{(A_2B_2) (A_1B_1)}^{T_{B_2B_1}} \geq 0$, as it is free, and $\lambda_i >0$. This is proves that $d_{A}(\mu_{(A_2B_2) \sep(A_1B_1)'} * \rho_{(A_1B_1)})^{T_{B_2}} \geq 0$, and that the set of states with positive partial transpose is closed under link product. In this case as well, we have shown that the set of free states satisfies the conditions of Theorem~\ref{th:cdrt}. The CDRT constructed coincides with the theory of NPT~\cite{Rai99b, APE03, IP05, MW08}.

\subsection{Example: Resource theory of non-negativity of quantum amplitudes}
In the resource theory of non-negativity of quantum amplitudes~\cite{JS22}, there is a fixed orthonormal basis $\Set{\Ket{j}_A}$ for every system $A$. For consistency, the orthonormal basis on composite systems $AB$ is $\Set{\Ket{jk}_{AB}}$, where $\Set{\Ket{j}_A}$ and $\Set{\Ket{k}_B}$ are the fixed orthonormal basis for systems $A$ and $B$, respectively. In the resource theory of non-negativity of quantum amplitudes, a normalized pure state $\Ket{\psi}_A$ is free if and only if, up to a global phase, it is a positive linear combination of elements of the fixed basis, i.e., if and only if there exists $\theta$ such that $\Ket{\psi}_A = e^{i\theta}\sum_{j}a_i \Ket{j}_A$, $a_j \geq 0$ for all $j$, and $\sum_j a^2_j =1$. If $\Ket{\psi}_A$ satisfies this condition we write $\Ket{\psi}_A \geq 0$. From this definition, it immediately follows that $\frac{1}{\sqrt{d_A}}\sum_{j}\Ket{jj}_{AA'}$ is a free pure state. In other words, the renormalized Choi state $\frac{1}{d_A}\Phi_{AA'}$ is free.

A mixed state is free, if it is a convex combination of free pure states, i.e., $\rho_A$ is free if and only if $\rho_A = \sum_\alpha p_\alpha\KBra{\psi_\alpha}{\psi_\alpha}_A$, where $\Ket{\psi_\alpha}_A \geq 0$, $p_\alpha \geq 0$ for all $j$, and $\sum_\alpha p_\alpha =1$. Note that a free mixed state expressed in the orthonormal basis has the form
\begin{equation}
	\sum_{\alpha, k, l}p_\alpha a^\alpha_k a^\alpha_l \KBra{k}{l}_A,
\end{equation}
where $p_\alpha \geq 0$, $a^\alpha_j \geq 0$, and $\sum_\alpha p_\alpha = \sum_j (a^\alpha_j)^2 =1$.

To show that the set of free states satisfies the conditions of Theorem~\ref{th:cdrt}, we need to show that $d_A \mu_{B\sep A'}*\rho_A$ is free whenever $\mu_{B\sep A'}$ and $\rho_A$ are free and $d_A \tr_B \mu_{BA} = \id_A$. Let $\mu_{BA} = \sum_{\alpha, j, k, l, m}{p_\alpha a^\alpha_{j,k}a^\alpha_{l,m}\KBra{jk}{lm}_{BA}}$, with $p_\alpha \geq 0$, $a^\alpha_{j,k} \geq 0$, $\sum_\alpha p_\alpha = \sum_{j,k}(a^\alpha_{j,k})^2 =1$, and let $\rho_A = \sum_{\beta,x,y}q_\beta b^\beta_xb^\beta_y \KBra{x}{y}_{A}$, with $q_\beta \geq 0$, $b^\beta_x \geq 0$ and $\sum_\beta q_\beta = \sum_x (b^\beta_x)^2 =1$. Then,
\begin{equation}
	\begin{aligned}
		 & d_A \mu_{B \sep A'} * \rho_A  =d_A \tr_A[\mu_{BA}(\id_B \otimes \rho_A^T)]                                                                          \\
		 & = d_A\sum_{\alpha, \beta,j,k,l,m,x,y}\tr_A [p_\alpha a^\alpha_{j,k}a^\alpha_{l,m}\KBra{jk}{lm}_{BA}                                                 \\
		 & \qquad(\id_B \otimes q_\beta b^\beta_xb^\beta_y \KBra{x}{y}_{A}) ]                                                                                  \\
		 & =d_A\sum_{\alpha, \beta,j,k,l,m}p_\alpha a^\alpha_{j,k}a^\alpha_{l,m} q_\beta b^\beta_k b^\beta_m \KBra{j}{l}_B                                     \\
		 & = d_A \sum_{\alpha, \beta,j,l} p_\alpha q_\beta \Bigl(\sum_k a^\alpha_{j,k}b^\beta_k\Bigr)\Bigl(\sum_m a^\alpha_{l,m} b^\beta_m\Bigr)\KBra{j}{l}_B.
	\end{aligned}
\end{equation}
This expression suggests the definition of the coefficients
\begin{equation}
	\begin{aligned}
		 & r_{\alpha, \beta} \coloneqq d_A p_\alpha q_\beta,                \\
		 & c_{j}^{\alpha,\beta} \coloneqq \sum_k a^\alpha_{j, k} b^\beta_k.
	\end{aligned}
\end{equation}
With such coefficients, the expression above becomes
\begin{equation}
	d_A \mu_{B \sep A'} * \rho_A = \sum_{\alpha, \beta, j, l} r_{\alpha, \beta}c^{\alpha, \beta}_j c^{\alpha, \beta}_l \KBra{j}{l}_B.
\end{equation}
However, these coefficients do not satisfy the desired conditions $\sum_{\alpha, \beta} r_{\alpha, \beta} = \sum_j (c^{\alpha, \beta}_j)^2 =1$. If one defines the renormalized coefficients
\begin{equation}
	\begin{aligned}
		 & s_{\alpha, \beta} \coloneqq r_{\alpha, \beta} \sum_{k} (c_k^{\alpha, \beta})^2,                      \\
		 & d_{j}^{\alpha, \beta} \coloneqq \frac{c_j^{\alpha, \beta}}{\sqrt{\sum_{k} (c_k^{\alpha, \beta})^2}},
	\end{aligned}
\end{equation}
it immediately follows that
\begin{equation}
	\begin{aligned}
		d_A \mu_{B \sep A'} * \rho_A & = \sum_{\alpha,\beta, j, l} r_{\alpha,\beta}c^{\alpha,\beta}_j c^{\alpha,\beta}_l \KBra{j}{l}_B  \\
		                             & = \sum_{\alpha,\beta, j, l} s_{\alpha,\beta}d^{\alpha,\beta}_j d^{\alpha,\beta}_l \KBra{j}{l}_B,
	\end{aligned}
\end{equation}
and
\begin{equation}
	\sum_{j}(d^{\alpha,\beta}_j)^2 = 1.
\end{equation}
What is left to show is that $\sum_{\alpha,\beta} s_{\alpha,\beta} =1$. To this end, notice that $d_A \tr_B \mu_{BA} = \id_A$ implies
\begin{equation}
	\begin{aligned}
		\delta_{k,l} & = \Braket{k|d_A \tr_B \mu_{BA}|l}_A                                                                                            \\
		             & = \sum_j d_A \Braket{jk|\mu_{BA}|jl}_{BA}                                                                                      \\
		             & =\sum_{j}d_A \Bra{jk}\Bigl(\sum_{\alpha, x, y, w, z}p_\alpha a^\alpha_{x,y}a^\alpha_{w,z}\KBra{xy}{wz}_{BA}\Bigr)\Ket{jl}_{BA} \\
		             & = \sum_{\alpha,j} d_A p_\alpha a^\alpha_{j,k}a^\alpha_{j, l}.
	\end{aligned}
\end{equation}Therefore
\begin{equation}
	\begin{aligned}
		\sum_{\alpha, \beta} s_{\alpha, \beta} & =\sum_{\alpha, \beta}r_{\alpha, \beta}\sum_j c_j^{\alpha, \beta} c_j^{\alpha, \beta}                       \\
		                                       & =\sum_{\alpha, \beta} d_A p_\alpha q_\beta \sum_{j,k,l} a^\alpha_{j,k}b^\beta_k a^{\alpha}_{j,l} b^\beta_l \\
		                                       & = \sum_{\beta,k,l}q_\beta b^\beta_k b^\beta_l \sum_{\alpha, j}d_A p_\alpha a^\alpha_{j, k} a^\alpha_{j,l}  \\
		                                       & = \sum_{\beta, k, l}q_\beta b^{\beta}_k b^\beta_l \delta_{k,l}                                             \\
		                                       & =\sum_{\beta}q_\beta \sum_k b^\beta_k b^\beta_k                                                            \\
		                                       & =\sum_{\beta}q_\beta                                                                                       \\
		                                       & =1.
	\end{aligned}
\end{equation}
The state $d_A \mu_{B \sep A'} * \rho_A$ can be written as a convex sum of positive pure vectors, and is therefore free. The CDRT generated from this set of pure states is the resource theory of non-negativity of quantum amplitudes.

\section{Optimization problems in CDRTs}\label{sec:opt-cdrt} In CDRTs, every optimization problem over the set of free channels can be converted into an optimization problem over the set of free states. That is,
\begin{equation}
	\max_{\mathcal{M}_{A \to B} \in \mathfrak{F}(A \to B)} f(\mathcal{M}_{A \to B})
\end{equation}
becomes
\begin{equation}
	\max_{\substack{\mu_{BA'} \in \mathfrak{F}(BA'),              \\d_{A'} \tr_B\mu_{BA'}=  \id_{A'}}} \tilde{f}(\mu_{BA'}).
\end{equation}
Here, $f$ is a function from the set of quantum channels to the $\mathbb{R}$, $\mathfrak{F}(A \to B)$ and $\mathfrak{F}(BA')$ are the sets of free channels from $A$ to $B$ and free states on $BA'$, and $\tilde{f}$ is the composition of $f$ with the inverse Choi isomorphism defined in Eq.~\eqref{eq:choi-inv}, and it satisfies $\tilde{f}(\mu_{BA'}) = f(\mathcal{M}_{A \to B})$ whenever $\mu_{BA'}$ is the renormalized Choi matrix of $\mathcal{M}_{A \to B}$. Moreover, if the set of free states is convex and $\tilde{f}$ is linear, the optimization problem can be expressed as a CLP (see Appendix~\ref{sec:clp} for details) or even as an SDP.

If one has a set of free states that does not satisfy the conditions of Theorem~\ref{th:cdrt}, one can still consider the CRNG operations. However, in this case, there is no universal way of expressing the constraints in optimization problems as a single constraint in terms of free states.

\section{Quantifying resources in CDRTs}\label{sec:rel-ent}
CDRTs are constructed from a given set of free states that satisfies the conditions of Theorem~\ref{th:cdrt}. Consequently, they inherit all the resource measures that depend only on the set of free states. Examples of such resource measures are the robustness, the generalized robustness, and all the measures based on quantum divergences (see Ref.~\cite{CG19} for a review). Among these, the max-relative entropy~\cite{Dat09} is particularly interesting in the context of CDRTs. It is defined as
\begin{equation}
	D_{\max}(\rho_A\Vert \sigma_A) = \log \min\Set{\lambda \mid \rho_A \leq \lambda\sigma_A},
\end{equation}
when $\supp \rho_A\subseteq\supp \sigma_A$ and $+\infty$ otherwise. The resource measure based on the max-relative entropy is
\begin{equation}
	D_{\max}^\mathfrak{F}(\rho_{A}) = \inf_{\omega_A \in \mathfrak{F}(A)} D_{\max}(\rho_A\Vert \omega_A).
\end{equation}

The extensions of these quantities to quantum channels are~\cite{CMW16, LKDW18}
\begin{equation}
	\begin{aligned}
		 & D_{\max}(\mathcal{M}_{A \to B}\Vert \mathcal{N}_{A \to B}) =                                                                    \\
		 & \quad\sup_{\rho_{RA}}D_{\max}(\mI_{R} \otimes \mM_{A \to B} (\rho_{RA})\Vert\mI_{R} \otimes \mN_{A \to B} (\rho_{RA})),         \\
		 & D_{\max}^\mathfrak{F}(\mM_{A \to B}) = \inf_{\mN_{A \to B} \in \mathfrak{F}(A \to B)}D_{\max}(\mM_{A \to B}\Vert\mN_{A \to B}).
	\end{aligned}
\end{equation}
Notably, the max-relative entropy for quantum channels satisfies~\cite{WBHK20}
\begin{equation}
	D_{\max}(\mM_{A \to B}\Vert \mN_{A \to B}) = D_{\max}(\mu_{BA'}\Vert \nu_{BA'}),
\end{equation}
where $\mu_{BA'}$ and $\nu_{BA'}$ are the renormalized Choi matrices of $\mM_{A \to B}$ and $\mN_{A \to B}$, respectively. In a CDRT, this property is particularly relevant because it allows us to write the optimization problem for $D_{\max}^\mathfrak{F}(\mM_{A\to B})$ as
\begin{equation}
	D_{\max}^\mathfrak{F}(\mM_{A \to B}) = \inf_{\substack{\omega_{BA'} \in \mathfrak{F}(BA'),              \\d_{A'} \tr_B\omega_{BA'}=  \id_{A'}}} D_{\max}(\mu_{BA'}\Vert \omega_{BA'}).
\end{equation}
Notice that, in every CDRT, $D_{\max}^\mathfrak{F}$ is finite for both states and channels because the set of free states always contains a full rank state, see Corollary~\ref{co:max-mix}. The same is not true for all resource theories.

As a direct consequence of the results in Ref.~\cite{WBHK20}, the resource measures based on the max-relative entropy satisfy the following inequality.
\begin{proposition}
	For all states $\rho_A$ and channels $\mM_{A \to B}$
	\begin{equation}\label{eq:chain-rule}
		D_{\max}^\mathfrak{F}(\mM_{A \to B}(\rho_A))  \leq D_{\max}^\mathfrak{F}(\rho_A) + D_{\max}^\mathfrak{F}(\mM_{A \to B}).
	\end{equation}
\end{proposition}
\begin{proof}
	We copy Eq. (52) and Eq. (85) of Ref.~\cite{WBHK20}, adapted to our notation, which will be useful for the next step. Eq. (52) is the definition of $D^A_{\max}$:
	\begin{equation}
		\begin{aligned}
			 & D^A_{\max}(\mM_{A \to B}\Vert \mN_{A \to B})  \coloneq                                                                       \\
			 & \sup_{\rho_{RA}, \sigma_{RA}}[D_{\max}(\mI_R \otimes \mM_{A \to B}(\rho_{RA})\Vert \mI_R\otimes  \mN_{A \to B}(\sigma_{RA})) \\
			 & \qquad- D_{\max}(\rho_{RA}\Vert \sigma_{RA})].
		\end{aligned}
	\end{equation}
	Eq. (85) shows the relation between $D^A_{\max}$ and $D_{\max}$:
	\begin{equation}
		D^A_{\max}(\mM_{A \to B}\Vert \mN_{A \to B}) = D_{\max}(\mM_{A \to B}\Vert \mN_{A \to B}).
	\end{equation}

	These equations imply that:
	\onecolumngrid
	\begin{equation}
		\begin{aligned}
			D_{\max}(\mM_{A \to B}\Vert \mN_{A \to B}) & = \sup_{\rho_{RA}, \sigma_{RA}}[D_{\max}(\mI_R \otimes \mM_{A \to B}(\rho_{RA})\Vert \mI_R\otimes  \mN_{A \to B}(\sigma_{RA}))- D_{\max}(\rho_{RA}\Vert \sigma_{RA})]                       \\
			                                           & \geq D_{\max}(\mI_R \otimes \mM_{A \to B}(\tau_R \otimes \rho_{A})\Vert \mI_R\otimes  \mN_{A \to B}(\tau_R \otimes \omega_A))- D_{\max}(\tau_R \otimes \rho_A\Vert \tau_R \otimes \omega_A) \\
			                                           & =  D_{\max}(\tau_R \otimes \mM_{A \to B}(\rho_{A})\Vert \tau_R\otimes  \mN_{A \to B}(\omega_A))- D_{\max}(\tau_R \otimes \rho_A\Vert \tau_R \otimes \omega_A)                               \\
			                                           & =D_{\max}(\mM_{A \to B}(\rho_A)\Vert \mN_{A \to B}(\omega_A)) - D_{\max}(\rho_A\Vert \omega_A),
		\end{aligned}
	\end{equation}
	\twocolumngrid
	where the last equality follows from the additivity of the max-relative entropy, and $\rho_A$, $\omega_A$, and $\tau_R$ are arbitrary quantum states. As a consequence, for every quantum state $\rho_A$ and $\omega_A$ and every channel $\mM_{A \to B}$ and $\mN_{A \to B}$, the following holds
	\begin{equation}
		\begin{aligned}
			 & D_{\max}(\mM_{A \to B}\Vert \mN_{A \to B}) +  D_{\max}(\rho_A\Vert \omega_A) \\
			 & \qquad \geq D_{\max}(\mM_{A \to B}(\rho_A)\Vert \mN_{A \to B}(\omega_A)) .
		\end{aligned}
	\end{equation}
	By taking the infimum on both sides over all free channels $\mN_{A \to B}$ and free states $\omega_A$, one obtains:
	\begin{equation}\label{eq:inf-sum}
		\begin{aligned}
			 & \inf_{\mN_{A \to B} \in \mathfrak{F}(A \to B)} D_{\max}(\mM_{A \to B}\Vert \mN_{A \to B}) \\
			 & \qquad+ \inf_{\omega_A \in \mathfrak{F}(A)} D_{\max}(\rho_A\Vert \omega_A)\geq            \\
			 & \inf_{\substack{\mN_{A \to B} \in \mathfrak{F}(A \to B),                                  \\ \omega_A \in \mathfrak{F}(A)}} D_{\max}(\mM_{A \to B}(\rho_A)\Vert \mN_{A \to B}(\omega_A)),
		\end{aligned}
	\end{equation}
	where the infimum on the left-hand side splits because $D_{\max}(\mM_{A \to B}\Vert \mN_{A \to B})$ does not depend on $\omega_A$ and $D_{\max}(\rho_A\Vert \omega_A)$ does not depend on $\mN_{A \to B}$. Moreover, since both $\mN_{A \to B}$ and $\omega_A$ are free, then $\mN_{A \to B}(\omega_A)$ is free as well, and therefore
	\begin{equation}\label{eq:inf-free}
		\begin{aligned}
			 & \inf_{\substack{\mN_{A \to B} \in \mathfrak{F}(A \to B),                                      \\ \omega_A \in \mathfrak{F}(A)}} D_{\max}(\mM_{A \to B}(\rho_A)\Vert \mN_{A \to B}(\omega_A))\\
			 & \qquad \geq\inf_{\omega_B \in \mathfrak{F}(B)} D_{\max}(\mM_{A \to B}(\rho_A)\Vert \omega_B).
		\end{aligned}
	\end{equation}
	Eq.~\eqref{eq:inf-sum} and Eq.~\eqref{eq:inf-free} imply
	\begin{equation}
		\begin{aligned}
			 & \inf_{\mN_{A \to B} \in \mathfrak{F}(A \to B)} D_{\max}(\mM_{A \to B}\Vert \mN_{A \to B})+                                                                  \\
			 & \inf_{\omega_A \in \mathfrak{F}(A)} D_{\max}(\rho_A\Vert \omega_A)  \geq \inf_{\omega_B \in \mathfrak{F}(B)} D_{\max}(\mM_{A \to B}(\rho_A)\Vert \omega_B).
		\end{aligned}
	\end{equation}
	This is equivalent to
	\begin{equation}
		D^{\mathfrak{F}}_{\max}(\mM_{A \to B}) + D^{\mathfrak{F}}_{\max}(\rho_A) \geq D^{\mathfrak{F}}_{\max}(\mM_{A \to B}(\rho_A)).
	\end{equation}
\end{proof}
Since in a CDRT $D^{\mathfrak{F}}_{\max}$ is always finite and computable with a CLP, this inequality, valid for all resource theories, is particularly meaningful in the context of CDRTs. When we interpret $D^\mathfrak{F}$ as the value of a resource, it implies that the value of the output of a quantum operation applied to a quantum state is always less than the sum of the values of the operation and the input state. Equivalently, for every $\rho_A$, $\sigma_B$, and $\mM_{A \to B}$ such that $\sigma_B = \mM_{A \to B}(\rho)$, the cost of preparing $\sigma_B$ is always less or equal than the cost of preparing $\rho_A$ and then converting it to $\sigma_B$ with the channel $\mM_{A \to B}$. An inequality similar to the one in Eq.~\eqref{eq:chain-rule} based on the quantum relative entropy and its regularization was derived in Ref.~\cite{FFRS20}.

\section{Resource conversion in CDRTs}

The formalism of CDRTs provides significant insights into all quantities connected with resource conversion or manipulation via free operations. Indeed at the core of every resource theory there is a preorder of resources defined as $\rho_A \succ \sigma_B$ if there exists a free channel $\mathcal{M}_{A \to B}$ such that $\sigma_B = \mathcal{M}_{A \to B}(\rho_A)$. In a CDRT, this condition is equivalent to requiring the existence of a free states $\mu_{BA}$, such that $\tr_B \mu_{BA}=\frac{1}{d_A} \id_A$ and $\sigma_B = d_A \mu_{B \sep A'} * \rho_A$.

Consequently, many quantities or results involving free channels can be redefined in terms of the link product of free states. A first example is the complete family of monotones defined in Ref.~\cite{GS20a} for convex and closed resource theories, which in the case of a convex and closed CDRT can be expressed as
\begin{equation}
	f_{\tau_B}(\rho_A) = \max_{\substack{\omega_{BA} \in \mathfrak{F}(BA),\\d_A \tr_B\omega_{BA}=  \id_A}}\tr_{BA}[\omega_{BA}(\tau_B\otimes \rho_A^T)].
\end{equation}
Being a complete family of monotones, it satisfies $f_{\tau_B}(\rho_A) \geq f_{\tau_B}(\sigma_B)$ for all density matrices $\tau_B$ if and only if $\rho_A$ can be converted to $\sigma_B$ with a free operation.

As a second example, we present conversion distances in CDRTs. These distances estimate how close to a target state an initial state can be converted using only free operations. In a CDRT, they satisfy
\begin{equation}
	\begin{aligned}
		D(\rho_A \to \sigma_B) & = \inf_{\mathcal{M}_{A \to B}\in \mf(A \to B)} D(\sigma_B, \mathcal{M}_{A \to B}(\rho_A)) \\
		                       & = \inf_{\substack{\mu_{BA'} \in \mf(BA'),                                                 \\ d_{A'}\tr_B\mu_{BA'} = \id_{A'}}} D(\sigma_B, d_A \mu_{B \sep A'}*\rho_A),
	\end{aligned}
\end{equation}
where $D(\cdot, \cdot)$ is a function on quantum states that satisfies the data processing inequality~\cite{CG19}.

\section{Conclusions}
In this work, we identified a key property of the resource theories of magic states, imaginarity, SEP, NPT, and non-negativity of quantum amplitudes: a channel is considered free if and only if its suitably renormalized Choi matrix is a free state. We presented a new way to construct a resource theory from the set of free states based on the Choi isomorphism, and we derived which conditions on the free states are necessary and sufficient for constructing a mathematically consistent Choi-defined resource theory. Our interest is driven by the significant advantages that such a property guarantees. These advantages stem from the fact that every question in Choi-defined resource theories can be naturally framed as a question involving only free states, which often have a well-understood structure. Indeed, we introduce several resource quantifiers that can all be computed with CLPs if the set of free states is convex and closed. Moreover, we showed that in any CDRT, free operations are all CRNG operations. This provides a new constructive way to define CRNG operations when the set of free states satisfies the conditions of Theorem~\ref{th:cdrt}.

A natural future direction is to explore Choi-defined \emph{dynamical} resource theories~\cite{GS20a, ZW19, LY20}, i.e., resource theories for channels, where a superchannel~\cite{Gou19, CdAP08} is free if and only if its Choi matrix is the Choi matrix of a free channel. Notably, such property has already been considered in the resource theories of dynamical entanglement~\cite{GS20} and magic channels~\cite{SG22}, where it led to interesting results. The characterization of Choi-defined dynamical resource theories is a particularly challenging problem because the techniques employed in this work cannot be easily translated to superchannels. Indeed, there are some subtleties in the definition of the Choi matrix of a superchannel. For example, in Ref.~\cite{Gou19}, several different Choi matrices are constructed from the same superchannel, which differ in the order in which systems are considered. As in the case studied in this article, an ordering of systems is mandatory if one wants to construct resource theories from Choi matrices, as inputs and outputs play different roles, especially if we want such a construction to be generalizable, without ambiguity, to even higher-order maps~\cite{BP19, Per17}.
\medskip{}

\begin{acknowledgments}
	The authors are grateful to Matt Wilson for the several in-depth discussions on the subject of this article and to Gilad Gour and Nathaniel Johnston for the interesting conversations about this research. E.~Z.\ acknowledges support from the Eric Milner’s Graduate Scholarship, the Alberta Graduate Excellence Scholarship (AGES), and the Eyes High International Doctoral Scholarship. C.\ M.\ S.\ acknowledges the support of the Natural Sciences and Engineering Research Council of Canada (NSERC) through the Discovery Grant “The power of quantum resources” RGPIN-2022-03025 and the Discovery Launch Supplement DGECR-2022-00119.
\end{acknowledgments}

\onecolumngrid
\appendix
\section{The Link Product and Swap Tensor Product}\label{sec:prod}

A very useful operation between Choi matrices is the link product, defined as
\begin{equation}
	N_{C \sep B'}*M_{B \sep A'} = \tr_{B'B}[(N_{CB'} \otimes M_{BA'})(\id_C \otimes \Phi_{B'B} \otimes \id_A)].
\end{equation}
Indeed, if $M_{B \sep A'}$ and $N_{C \sep B'}$ are the Choi matrices of the quantum channels $\mathcal{M}_{A \to B}$ and $\mathcal{N}_{B \to C}$, respectively, then the matrix $T_{C \sep A'} = N_{C \sep B'}*M_{B \sep A'}$ is the Choi matrix of the quantum channel $\mathcal{T}_{A \to C} = \mathcal{N}_{B \to C} \circ \mathcal{M}_{A \to B}$~\cite{CdAP09}. It is worth noticing that if $\mathcal{M}_{A \to B}$ is a preparation channel, i.e., $A = \mathbb{C}$ and $\mathcal{M}_{\mathbb{C} \to B}(1) = \rho_B$, then $M_{B \sep \mathbb{C}'} = \rho_B$ and $N_{C \sep B'}*M_{B \sep \mathbb{C}'} = \mathcal{N}_{B \to C}(\rho_B)$.

The link product is the operation between Choi matrices associated with the sequential composition of channels. Now, we focus on parallel composition. Let $\mathcal{M}_{A \to B}$ and $\mathcal{N}_{C \to D}$ be quantum channels, and let $M_{B \sep A'}$ and $N_{D \sep C'}$ be the Choi matrices associated with them. Let $\mathcal{T}_{AC \to BD} = \mathcal{M}_{A \to B} \otimes \mathcal{N}_{C \to D}$. The Choi matrix of the tensor product of channels is often either overlooked or claimed to be equal to $M_{B \sep A'} \otimes N_{D \sep C'}$. The latter approach has its merits if the order of the systems is not relevant. However, when it is crucial to keep track of input and output systems, like in this work, we expect the Choi matrix of $\mathcal{T}_{AC \to BD}$ to be a matrix $T_{BD \sep A'C'}$ which is not equal to $M_{B \sep A'} \otimes N_{D \sep C'}$ because the order of the systems is not the same. We define here an operation $\boxtimes$, which we call \emph{swap tensor product}, that takes care of the order of systems:
\begin{equation}
	M_{B \sep A'} \boxtimes N_{D \sep C'} \coloneqq (\mathcal{I}_{B} \otimes \mathcal{S}_{A'D \to DA'} \otimes \mathcal{I}_{C'})(M_{B \sep A'} \otimes N_{D \sep C'}).
\end{equation}

We show now that $M_{B \sep A'} \boxtimes N_{D \sep C'}$ is the Choi matrix $T_{BD \sep A'C'}$ of $\mathcal{T}_{AC \to BD} = \mathcal{M}_{A \to B} \otimes \mathcal{N}_{C \to D}$.
\begin{equation}\label{eq:swap-alg}
	\begin{aligned}
		T_{BD \sep A'C'} & = (\mathcal{T}_{AC \to BD} \otimes \mathcal{I}_{A'C'})(\Phi_{ACA'C'})                                                                                                                                                       \\
		                 & =[(\mathcal{M}_{A \to B} \otimes \mathcal{N}_{C \to D} \otimes \mathcal{I}_{A'C'}) \circ (\mathcal{I}_{A} \otimes \mathcal{S}_{A'C \to CA'} \otimes \mathcal{I}_{C'})](\Phi_{AA'} \otimes \Phi_{CC'})                       \\
		                 & =[(\mathcal{I}_{B} \otimes \mathcal{S}_{A'D \to DA'} \otimes \mathcal{I}_{C'})\circ (\mathcal{M}_{A \to B} \otimes \mathcal{I}_{A'} \otimes \mathcal{N}_{C \to D} \otimes \mathcal{I}_{C'})](\Phi_{AA'} \otimes \Phi_{CC'}) \\
		                 & = (\mathcal{I}_{B} \otimes \mathcal{S}_{A'D \to DA'} \otimes \mathcal{I}_{C'})(M_{B \sep A'} \otimes N_{D \sep C'})                                                                                                         \\
		                 & =M_{B \sep A'} \boxtimes N_{D \sep C'}.
	\end{aligned}
\end{equation}

\section{Conic linear programs and semidefinite programs}\label{sec:clp}
In Section~\ref{sec:opt-cdrt}, we have shown that in a CDRT, an optimization problem over free channels can be expressed as an optimization problem over free states. That is,
\begin{equation}
	\max_{\mathcal{M}_{A \to B} \in \mathfrak{F}(A \to B)} f(\mathcal{M}_{A \to B})
\end{equation}
becomes
\begin{equation}
	\max_{\substack{\mu_{BA'} \in \mathfrak{F}(BA'),              \\d_{A'} \tr_B\mu_{BA'}=  \id_{A'}}} \tilde{f}(\mu_{BA'}).
\end{equation}
As a reminder, $f$ is a function from the set of quantum channels to the $\mathbb{R}$, $\mathfrak{F}(A \to B)$ and $\mathfrak{F}(BA')$ are the sets of free channels from $A$ to $B$ and free states on $BA'$, and $\tilde{f}$ is the composition of $f$ with the inverse Choi isomorphism defined in Eq.~\eqref{eq:choi-inv}, and it satisfies $\tilde{f}(\mu_{BA'}) = f(\mathcal{M}_{A \to B})$ whenever $\mu_{B \sep A'}$ is the renormalized Choi matrix of $\mathcal{M}_{A \to B}$.  As in the case of the max-relative entropy (Section~\ref{sec:rel-ent}), it is often possible to find simpler expressions for $\tilde{f}$ that do not rely on the inverse Choi isomorphism.

If the set of free states is convex and $\tilde{f}$ is linear, the optimization problem can be expressed as a conic linear program (CLP) or even as a semidefinite program (SDP). Sometimes, these optimization problems are hard to solve when the dimension of the systems involved is large, for example, in the case of separable entanglement~\cite{Gur03, HO23}. However, they are still useful for systems of low dimensions.

Here we present the techniques used to write CLPs and provide some examples. In this article, we use the notation presented in Ref.~\cite{Gou24}, which is particularly useful in quantum information. In our setting, a \emph{primal} conic linear program is an optimization problem that can be expressed as
\begin{equation}\label{eq:clp}
	\begin{aligned}
		\text{Find}       &  &  & \alpha\coloneqq \inf \tr[X H_1]         \\
		\text{subject to} &  &  & \mathcal{N}(X) - H_2 \in \mathfrak{K}_2 \\
		                  &  &  & X \in \mathfrak{K}_1
	\end{aligned}
\end{equation}
where $\mathfrak{K}_1$ is a closed convex cone in $V_1$, a subspace of the Hermitian matrix on a Hilbert space $\mathcal{H}_1$. Similarly, $\mathfrak{K}_2 \subseteq V_2 \subseteq \herm(\mathcal{H}_2)$. Moreover, $\mN: V_1 \to V_2$ is a liner map, and $H_1 \in V_1$, $H_2 \in V_2$. We say that $X \in \mathfrak{K}_1$ is a \emph{feasible} solution if $\mathcal{N}(X) - H_2 \in \mathfrak{K}_2$. If no feasible solution exists, then $\alpha$ is set to $+\infty$. A feasible solution is \emph{optimal} if $\tr[X H_1] = \alpha$.

The \emph{dual} of the conic linear program in Eq.~\eqref{eq:clp} is
\begin{equation}
	\begin{aligned}
		\text{Find}       &  &  & \beta\coloneqq \sup \tr[Y H_2]                    \\
		\text{subject to} &  &  & H_1 - \mathcal{N}^\dagger(Y) \in \mathfrak{K}^*_1 \\
		                  &  &  & Y \in \mathfrak{K}^*_2,
	\end{aligned}
\end{equation}
where $\mathfrak{K}^*$ is the dual cone of $\mathfrak{K}$, and $\mathcal{N}^\dagger$ is the adjoint of $\mathcal{N}$. As in the case above, $Y \in \mathfrak{K}_2^*$ is a dual feasible solution if $H_1 - \mathcal{N}^\dagger(Y) \in \mathfrak{K}_1^*$. If no dual feasible solution exists, then $\beta = -\infty$. A dual feasible solution is optimal if $\tr[YH_2]=\beta$. The dual CLP is often relevant because it provides a lower bound for the primal CLP, that is, $\alpha \geq \beta$. This relation is called weak duality. Strong duality happens if $\alpha = \beta$.

In this section, $V_1\subseteq \herm(\mathcal{H}_1)$ will always be the vector space spanned by block-diagonal Hermitian matrices, where the dimensions of the blocks are fixed. We will denote such matrices as $X = (X_1, \dots, X_n)$. Each of the matrices $X_i$ is a Hermitian matrix on an orthogonal subspace $A_i$ of $\mathcal{H}_1$. As a consequence, $V_1 \cong \oplus_{i =1}^n \herm{A_i}$. An analogous statement holds for $V_2$.

\subsection{CLPs for a Complete Family of Monotones in CDRTs}
We start with the complete family of monotones $f_{\tau_B}(\rho_A) = \max_{\substack{\omega_{BA} \in \mathfrak{F}(BA),\\d_A \tr_B\omega_{BA}=  \id_A}}\tr_{BA}[\omega_{BA}(\tau_B\otimes \rho_A^T)]$, where $\tau_B$ and $\rho_A$ are density matrices. To compute one of such monotones with a conic linear program, the first step is to define the cone generated by the free states.
\begin{equation}
	\mathfrak{K}_1 = \Set{(\lambda, \lambda \omega_{BA}) \mid \lambda \geq 0, \omega_{BA} \in \mathfrak{F}(BA)} \subseteq V_1 \cong \mathbb{R}\oplus \herm(BA).
\end{equation}
The conditions $\omega_{BA} \in \mathfrak{F}(BA)$ and $d_A \tr_B\omega_{BA}=  \id_A$ are equivalent to $(
	\lambda,\lambda \omega_{BA} ) \in \mathfrak{K}_1$ and $d_A \tr_B\lambda\omega_{BA} -  \lambda\id_A = 0$. The last condition is implemented in a CLP by defining $\mathcal{N}[(x,X_{BA})] = d_A \tr_BX_{BA} -  x\id_A $, $H_2 = 0$, and $\mathfrak{K}_2 = \Set{0_A} \subseteq V_2 = \herm(A)$. The conic linear program is
\begin{equation}\label{eq:mon-primal}
	\begin{aligned}
		\text{Find}       &  &  & -f_{\tau_B}(\rho_A)\coloneqq \min \tr[(x,X_{BA})(0,-\tau_B \otimes \rho^T_A)] \\
		\text{subject to} &  &  & \mathcal{N}[(x,X_{BA})] = d_A \tr_BX_{BA} -  x\id_A= 0_A                      \\
		                  &  &  & (x,X_{BA}) \in \mathfrak{K}_1.
	\end{aligned}
\end{equation}

Now, we turn our attention to the dual of the CLP above. Since $H_2= 0_A$ then $\beta$ (the result of the dual CLP) is $0$ if a dual feasible solution exists or $-\infty$ if no dual feasible solution exists. In the former case, $-f_{\tau_B}(\rho_A) = \alpha \geq \beta = 0$, and this would imply that $f_{\tau_B}(\rho_A) = 0$ because $f_{\tau_B}(\rho_A)$ is positive by definition. Interestingly, in every CDRT, we can exclude this case because $\frac{1}{d_Ad_B} \id_{BA}$ is free, and therefore
\begin{equation}
	f_{\tau_B}(\rho_A) = \max_{\substack{\omega_{BA} \in \mathfrak{F}(BA),\\d_A \tr_B\omega_{BA}=  \id_A}}\tr_{BA}[\omega_{BA}(\tau_B\otimes \rho_A^T)] \geq \frac{1}{d_Ad_B}\tr_{BA}[\tau_B\otimes \rho_A^T] = \frac{1}{d_Bd_A}
\end{equation}

This implies that in every CDRT $\beta = -\infty$, and weak duality does not provide a meaningful bound for $f_{\tau_B}(\rho_A)$. However, writing the dual of the CLP in Eq.~\eqref{eq:mon-primal} is still useful. The first step is to characterize the dual cones of $\mathfrak{K}_1$ and $\mathfrak{K}_2$:
\begin{equation}\label{eq:k1-star}
	\begin{aligned}
		\mathfrak{K}_1^* & = \Set{(x, X_{BA}) \in V_1 \mid \Braket{(x, X_{BA})|(\lambda, \lambda \omega_{BA})}_{HS} \geq 0,\, \forall \lambda \geq 0,\, \forall \omega_{BA} \in \mathfrak{F}_{BA}} \\
		                 & =\Set{(x, X_{BA}) \in V_1 \mid \Braket{(x, X_{BA})|(1, \omega_{BA})}_{HS} \geq 0,\, \forall \omega_{BA} \in \mathfrak{F}_{BA}},                                         \\
		\mathfrak{K}_2^* & = \Set{X_{A} \in V_2\mid \Braket{X_A|0_A}_{HS} \geq 0} = V_2,
	\end{aligned}
\end{equation}
where $\Braket{\cdot|\cdot}_{HS}$ is the Hilbert-Schmidt inner product defined as $\Braket{X|Y}_{HS} = \tr[X^\dagger Y]$, i.e., $\Braket{(x, X_{BA})|(y, Y_{BA})}_{HS} = xy + \tr_{BA}[X_{BA}Y_{BA}]$. $\mathcal{N}^\dagger : V_2 \to V_1$ is the unique map that satisfies
\begin{equation}\label{eq:adjoint}
	\Braket{(x, X_{BA})|\mathcal{N}^\dagger(X_A)}_{HS} = \Braket{\mathcal{N}[(x, X_{BA})]|X_A}_{HS}
\end{equation}
for all $(x, X_{BA}) \in V_1 \cong \mathbb{R} \oplus \herm(BA)$ and $X_A \in V_2 = \herm(A)$. We denote the components of $\mathcal{N}^\dagger$ with $\mathcal{N}^\dagger_1$ and $\mathcal{N}^\dagger_2$, that is, $\mathcal{N}^\dagger(X_A) = (\mathcal{N}^\dagger_1(X_A), \mathcal{N}^\dagger_2(X_A))$. To find the expression of $\mathcal{N}^\dagger$, we evaluate Eq.~\eqref{eq:adjoint} for each member of the basis $\Set{(1, 0_{BA})} \cup \Set{(0, \eta^k_B \otimes \eta^j_A)}$ of $V_1$, where $\Set{\eta^j_A}$ and $\Set{\eta^k_B}$ are orthonormal bases for $\herm(A)$ and $\herm(B)$, respectively. Moreover, we denote with $[N^\dagger_2(X_A)]^{kj}$, the components of $N^\dagger_2(X_A)$ in the $\Set{\eta^k_B \otimes \eta^j_A}$ basis, that is $\mathcal{N}^\dagger_2(X_A) = \sum_{j, k} [N^\dagger_2(X_A)]^{kj} \eta^k_B \otimes \eta^j_A$. Similarly, we denote with $X_A^j$ the components of $X_A$ in the $\Set{\eta^j_A}$ basis.
\begin{equation}
	\begin{aligned}
		 & \mathcal{N}_1^\dagger(X_A) &  & = \Braket{(1, 0_{BA})|(\mathcal{N}^\dagger_1(X_A), \mathcal{N}^\dagger_2(X_A))}_{HS}                    \\
		 &                            &  & = \Braket{(1, 0_{BA})|\mathcal{N}^\dagger(X_A)}_{HS}                                                    \\
		 &                            &  & = \Braket{\mathcal{N}[(1, 0_{BA})]|X_A}_{HS}                                                            \\
		 &                            &  & = -\Braket{\id_A|X_A}_{HS}                                                                              \\
		 &                            &  & = -\tr_A X_A,                                                                                           \\
		 & [N^\dagger_2(X_A)]^{kj}    &  & = \Braket{(0, \eta^k_B \otimes \eta^j_A)|(\mathcal{N}^\dagger_1(X_A), \mathcal{N}^\dagger_2(X_A))}_{HS} \\
		 &                            &  & =\Braket{(0, \eta^k_B \otimes \eta^j_A)|\mathcal{N}^\dagger(X_A)}_{HS}                                  \\
		 &                            &  & = \Braket{\mathcal{N}[(0, \eta^k_B \otimes \eta^j_A)]|X_A}_{HS}                                         \\
		 &                            &  & = \Braket{d_A \eta^j_A \tr_B\eta^k_B|X_A}_{HS}                                                          \\
		 &                            &  & =d_A X_A^j \tr_B \eta^k_B.
	\end{aligned}
\end{equation}
Therefore,
\begin{equation}
	\mathcal{N}^\dagger_2(X_A) = \sum_{j, k} [N^\dagger_2(X_A)]^{kj} \eta^k_B \otimes \eta^j_A = d_A \sum_i  (\tr_B\eta^k_B)\eta^k_B \otimes \sum_j X_A^j \eta^j_A = d_A  \id_B\otimes X_A,
\end{equation}
where $\sum_i (\tr_B\eta^k_B)\eta^k_B = \id_B$ because $\tr_B (\eta^k_B \id_B)$ are the components of $\id_B$ in the $\Set{\eta^k_B}$ basis.

Now, we can write the dual CLP as
\begin{equation}
	\begin{aligned}
		\text{Find}       &  &  & \beta\coloneqq\max 0                                                                                                     \\
		\text{subject to} &  &  & H_1 - \mathcal{N}^\dagger(X_A) = (0, -\tau_B \otimes \rho_A^T) - (-\tr_A X_A, d_A \id_B \otimes X_A)\in \mathfrak{K}^*_1 \\
		                  &  &  & X_A \in \herm(A),
	\end{aligned}
\end{equation}
where $\mathfrak{K}^*_1 = \Set{(x, X_{BA}) \in V_1 \mid \Braket{(x, X_{BA})|(1, \omega_{BA})}_{HS} \geq 0,\, \forall \omega_{BA} \in \mathfrak{F}_{BA}}$ (see Eq.~\eqref{eq:k1-star}). The condition $H_1 - \mathcal{N}^\dagger(X_A) \in \mathfrak{K}^*_1$ is equivalent to:
\begin{equation}
	\tr_A X_A - \tr_{BA}[(\tau_B \otimes \rho_A^T) \omega_{BA}] +d_A \tr_{BA}[(\id_B \otimes X_A)\omega_{BA}] \geq 0, \quad \forall \omega_{BA} \in \mathfrak{F}(BA).
\end{equation}
It is worth checking that no dual feasible solution exists for any CDRT. Indeed, in a CDRT, $\omega_{BA} = \frac{1}{d_A d_B} \id_{BA} \in \mathfrak{F}(A)$, and for all $X_A$ we have that
\begin{equation}
	\begin{aligned}
		 & \tr_A X_A - \tr_{BA}[(\tau_B \otimes \rho_A^T) \omega_{BA}] +d_A \tr_{BA}[(\id_B \otimes X_A)\omega_{BA}]           \\
		 & \qquad = \tr_{A}X_A -\frac{1}{d_A d_B} \tr_{BA}[\tau_B \otimes \rho_A^T] - \frac{1}{d_B}\tr_{BA}[\id_B \otimes X_A] \\
		 & \qquad = - \frac{1}{d_B d_A}.
	\end{aligned}
\end{equation}
Therefore, $H_1 - \mathcal{N}^\dagger(X_A)\notin \mathfrak{K}_1^*$ for all $X_A \in \herm(A)$, and no dual feasible solution exists.

\subsection{CLPs for the Max-Relative Entropy of a State}
With the same techniques used in the previous subsection, we write the conic linear program to compute ${D_{\max}^\mathfrak{F}(\rho_A)}$. First, we notice that
\begin{equation}
	2^{D_{\max}^\mathfrak{F}(\rho_A)} = \min_{\substack{\omega_A \in \mathfrak{F}(A)\\ \rho_A \leq \lambda \omega_A}} \lambda= \min_{\substack{\omega_A \in \mathfrak{F}(A)\\ \lambda \omega_A- \rho_A \geq 0}} \lambda.
\end{equation}
Second, we define the cone
\begin{equation}
	\mathfrak{K}_1 = \Set{(\lambda,\lambda \omega_A)\mid \lambda \geq 0, \omega_A \in \mathfrak{F}(A)} \subseteq V_1 \cong \mathbb{R} \oplus \herm(A),
\end{equation}
and the cone $\mathfrak{K}_2$ as the cone of positive semidefinite matrices on $A$, i.e., $\mathfrak{K}_2 = \herm_+(A) \subseteq V_2 = \herm(A)$. Lastly, $\mN[(x,X_A)] = X_A$ and $H_2 = \rho_A$. The conic linear program for the max relative entropy is
\begin{equation}
	\begin{aligned}
		\text{Find}       &  &  & 2^{D_{\max}^\mathfrak{F}(\rho_A)}\coloneqq \inf \tr[(1,0_A)(x,X_A)]  = \inf x \\
		\text{subject to} &  &  & \mN[(x,X_A)] - H_2 = X_A - \rho_A \geq 0,                                     \\
		                  &  &  & (x,X_A) \in \mathfrak{K}_1.
	\end{aligned}
\end{equation}

As before, to write the dual CLP, we start with the dual cones:
\begin{equation}\label{eq:k1-star-2}
	\begin{aligned}
		 & \mathfrak{K}_1^* = \Set{(x, X_A) \in V_1\mid \Braket{(x, X_A)|(1, \omega_A)}_{HS} \geq 0, \, \forall \omega_A \in \mathfrak{F}(A)}, \\
		 & \mathfrak{K}_2^* = (\herm_+(A))^* = \herm_+(A).
	\end{aligned}
\end{equation}

The components of the adjoint map are
\begin{equation}
	\begin{aligned}
		 & \mathcal{N}_1^\dagger(X_A) &  & = \Braket{(1, 0_{A})|(\mathcal{N}^\dagger_1(X_A), \mathcal{N}^\dagger_2(X_A))}_{HS}    \\
		 &                            &  & = \Braket{(1, 0_{A})|\mathcal{N}^\dagger(X_A)}_{HS}                                    \\
		 &                            &  & = \Braket{\mathcal{N}[(1, 0_A)]|X_A}_{HS}                                              \\
		 &                            &  & = -\Braket{0_A|X_A}_{HS}                                                               \\
		 &                            &  & = 0 ,                                                                                  \\
		 & [N^\dagger_2(X_A)]^{j}     &  & = \Braket{(0, \eta^j_A)|(\mathcal{N}^\dagger_1(X_A), \mathcal{N}^\dagger_2(X_A))}_{HS} \\
		 &                            &  & =\Braket{(0, \eta^j_A)|\mathcal{N}^\dagger(X_A)}_{HS}                                  \\
		 &                            &  & = \Braket{\mathcal{N}[(0, \eta^j_A)]|X_A}_{HS}                                         \\
		 &                            &  & = \Braket{\eta^j_A |X_A}_{HS}                                                          \\
		 &                            &  & =X_A^j,
	\end{aligned}
\end{equation}
where, as before, $\Set{\eta^j_A}$ is an orthonormal basis for $\herm(A)$. Thus, $\mathcal{N}^\dagger(X_A) = (0, X_A)$. The dual CLP is
\begin{equation}
	\begin{aligned}
		\text{Find}       &  &  & \beta\coloneqq \sup \tr[X_A\rho_A]                               \\
		\text{subject to} &  &  & H_1 - \mathcal{N}^\dagger(X_A) = (1,  -X_A) \in \mathfrak{K}^*_1 \\
		                  &  &  & X_A \geq 0,
	\end{aligned}
\end{equation}
where $\mathfrak{K}_1^* = \Set{(x, X_A) \in V_1\mid \Braket{(x, X_A)|(1, \omega_A)}_{HS} \geq 0, \, \forall \omega_A \in \mathfrak{F}(A)}$ (see Eq.~\eqref{eq:k1-star-2}). The condition $(1, -X_A) \in \mathfrak{K}_1^*$ is equivalent to $\tr_A[X_A \omega_A] \leq 1$ for all $\omega_A \in \mathfrak{F}(A)$. In this case, a feasible solution exists: The zero matrix $0_A$ is a feasible solution.

\subsection{CLPs for the Max-Relative Entropy of a Channel in CDRTs}
In this section, we compute the CLP for
\begin{equation}
	2^{D_{\max}^{\mathfrak{F}}(\mM_{A \to B})} = \min_{\substack{\omega_{BA'} \in \mf(BA'),\\d_{A'} \tr_B \omega_{BA'} = \id_{A'},\\ \mu_{BA'} \leq \lambda\omega_{BA'}}} \lambda,
\end{equation}
where $\mu_{B\sep A'}$ is the renormalized Choi matrix of $\mathcal{M}_{A \to B}$. Once again, we define
\begin{equation}
	\mathfrak{K}_1 = \Set{(\lambda,\lambda \omega_{BA'})\mid \lambda \geq 0, \omega_{BA'} \in \mathfrak{F}(BA')} \subseteq V_1 \cong \mathbb{R} \oplus \herm(BA').
\end{equation}
The linear transformation $\mN: \mathbb{R} \oplus \herm(BA') \to \herm(BA') \oplus \herm(A')$ is defined as the map that sends $(x,X_{BA'})$ into $(X_{BA'}, d_{A'}\tr_B X_{BA'} - x\id_{A'})$. Moreover, we define $H_2 = (\mu_{BA'},0_{A'})$, and

\begin{equation}
	\mathfrak{K}_2 = \Set{(X_{BA'},0_{A'})\mid X_{BA'} \geq 0} \subseteq V_2 \cong \herm(BA') \oplus \herm(A').
\end{equation}
The conic linear program is
\begin{equation}
	\begin{aligned}
		\text{Find}       &  &  & 2^{D_{\max}^\mathfrak{F}(\mM_{A \to B})}\coloneqq \inf \tr[(1,0_{BA'})(x,X_{BA'})]= \inf x              \\
		\text{subject to} &  &  & \mathcal{N}[(x,X_{BA'})] - H_2=(X_{BA'} - \mu_{BA'},d_{A'}\tr_BX_{BA'} - x\id_{A'}) \in \mathfrak{K}_2, \\
		                  &  &  & (x,X_{BA'}) \in \mathfrak{K}_1.
	\end{aligned}
\end{equation}

To write the dual CLP, we first need the dual cones:
\begin{equation}\label{eq:k1-star-3}
	\begin{aligned}
		\mathfrak{K}_1^* & = \Set{(x, X_{BA'}) \in V_1\mid \Braket{(x|X_{BA'})|(1, \omega_{BA'})}_{HS} \geq 0, \, \forall \omega_{BA'} \in \mathfrak{F}(BA')}, \\
		\mathfrak{K}_2^* & = \Set{(X_{BA'}, Y_{A'}) \in V_2 \mid \Braket{(X_{BA'}, Y_{A'})|(Z_{BA'}, 0_{A'})}_{HS} \geq 0,\,\forall Z_{BA'} \geq 0}            \\
		                 & =\Set{(X_{BA'}, Y_{A'}) \in V_2 \mid X_{BA'} \geq 0}.
	\end{aligned}
\end{equation}'

With the techniques introduced in the previous subsections, we compute the components of $\mathcal{N}^\dagger$.
\begin{equation}
	\begin{aligned}
		 & \mathcal{N}_1^\dagger[(X_{BA'}, Y_{A'})] &  & = \Braket{(1, 0_{BA'})|(\mathcal{N}^\dagger_1[(X_{BA'}, Y_{A'})], \mathcal{N}^\dagger_2[(X_{BA'}, Y_{A'})])}_{HS}                      \\
		 &                                          &  & = \Braket{(1, 0_{BA'})|\mathcal{N}^\dagger[(X_{BA'}, Y_{A'})]}_{HS}                                                                    \\
		 &                                          &  & = \Braket{\mathcal{N}[(1, 0_{BA'})]|(X_{BA'}, Y_{A'})}_{HS}                                                                            \\
		 &                                          &  & = -\Braket{(0_{BA'}, \id_{A'})|(X_{BA'}, Y_{A'})}_{HS}                                                                                 \\
		 &                                          &  & = -\tr_{A'} Y_{A'}    ,                                                                                                                \\
		 & \{N^\dagger_2[(X_{BA'}, Y_{A'})]\}^{kj}  &  & = \Braket{(0, \eta^k_B \otimes \eta^j_{A'})|(\mathcal{N}^\dagger_1[(X_{BA'}, Y_{A'})], \mathcal{N}^\dagger_2[(X_{BA'}, Y_{A'})])}_{HS} \\
		 &                                          &  & =\Braket{(0, \eta^k_B \otimes \eta^j_{A'})|\mathcal{N}^\dagger[(X_{BA'}, Y_{A'})]}_{HS}                                                \\
		 &                                          &  & = \Braket{\mathcal{N}[(0, \eta^k_B \otimes \eta^j_{A'})]|(X_{BA'}, Y_{A'})}_{HS}                                                       \\
		 &                                          &  & = \Braket{(\eta^k_B \otimes \eta^j_A, d_{A'} (\tr_B\eta^k_B)  \eta^j_{A'}) |(X_{BA'}, Y_{A'})}_{HS}                                    \\
		 &                                          &  & =X_{BA'}^{kj} + d_{A'} Y_{A'}^j(\tr_B \eta^k_B).
	\end{aligned}
\end{equation}
Thus, $N^\dagger[(X_{BA'}, Y_{A'})] = (-\tr_{A'} Y_{A'}, X_{BA'} + d_{A'} \id_B \otimes Y_{A'})$. The dual linear program is
\begin{equation}
	\begin{aligned}
		\text{Find}       &  &  & \beta\coloneqq \sup \tr[(X_{BA'}, Y_{A'})(\mu_{BA'}, 0_{A'})]  = \tr_{BA'}[X_{BA'} \mu_{BA'}]                                                \\
		\text{subject to} &  &  & H_1 - \mathcal{N}^\dagger[(X_{BA'}, Y_{A'})] = (1, 0_{BA'}) - (-\tr_{A'} Y_{A'}, X_{BA'} + d_{A'} \id_B \otimes Y_{A'}) \in \mathfrak{K}^*_1 \\
		                  &  &  & X_{BA'} \geq 0,\, Y_{A'} \in \herm(A'),
	\end{aligned}
\end{equation}
where $\mathfrak{K}_1^* = \Set{(x, X_{BA'}) \in V_1\mid \Braket{(x|X_{BA'})|(1, \omega_{BA'})}_{HS} \geq 0, \, \forall \omega_{BA'} \in \mathfrak{F}(BA')}$ (see Eq.~\eqref{eq:k1-star-3}). The condition $(1 + \tr_{A'} Y_{A'}, -X_{BA'} - d_{A'} \id_B \otimes Y_{A'}) \in \mathfrak{K}_1^*$ is equivalent to $\tr_{BA'}[X_{BA'}\omega_{BA'}] + d_{A'}\tr_{BA'}[(\id_B \otimes Y_{A'})\omega_{BA'}] \leq 1 + \tr_{A'} Y_{A'}$ for all $\omega_{BA'} \in \mathfrak{F}(BA')$. It is straightforward to see that $(0_{BA'}, 0_{A'})$ is a dual feasible solution.

\section{A Diagrammatic Proof of Theorem~\ref{th:cdrt}}
String diagrams have been frequently used in quantum information~\cite{CK10, CHK14, CDKW15, CFS16, KU17, KU19, SSC21}, and they provide a useful representation of the Choi isomorphism, from which many of its properties follow naturally. Indeed, we first proved all our results diagrammatically, and later we converted the diagrammatic proofs into algebraic ones. In this section, we introduce the string diagram formalism and provide a diagrammatic proof of Theorem~\ref{th:cdrt}.

\subsection{A Diagrammatic Approach to Quantum Information}
Quantum information processes can be described using string diagrams~\cite{CK10, CHK14, CDKW15, CFS16, KU17, KU19, SSC21}. These diagrams are not only a graphical representation of quantum processes but can also be used for rigorous proofs. Indeed, one can associate an algebraic expression with every diagram, and a diagram with every algebraic expression. Moreover, diagrams can be manipulated as algebraic expressions.

The fundamental component of string diagrams is the diagrammatic symbol for quantum channels. A channel $\mM_{A \to B}$ is represented in a diagram as a box:
\begin{center}
	\begin{tikzpicture}
		\node[shape = rectangle, minimum height=20pt, draw] (channel) {$\mathcal{M}$};
		\coordinate[left=of channel] (A);
		\coordinate[right=of channel] (B) {};
		\draw [-] (channel) to node [above] {$A$} (A) ;
		\draw [-] (channel) to node [above] {$B$} (B) ;
	\end{tikzpicture}\,.
\end{center}
This symbol represents the channel acting on an input quantum system $A$ and producing an output quantum system $B$. When a channel has $n$ input
systems and $m$ output systems it is represented with $n$ input wires and $m$ output wires. There are two special types of channels. The first is the preparation channel, which has no quantum input, represented with no wire, and outputs a state $\rho_A$. The symbol for this channel is
\begin{center}
	\begin{tikzpicture}
		\node[shape = rounded rectangle, rounded rectangle right arc = none, minimum height = 20, draw] (state) {$\rho$};
		\coordinate[right=of state] (end);
		\draw [-] (state) to node [above] {$A$} (end) ;
	\end{tikzpicture}.
\end{center}
Since there is a one-to-one correspondence between preparation channels and quantum states, we say that this is the symbol for the state $\rho_A$. With a slight abuse of notation, we will use this symbol to represent positive semi-definite matrices too. However, these two cases can easily be distinguished because we denote states with lowercase Greek letters, e.g., $\rho$, and matrices with capital Latin letters, e.g., $M$.

The second special type of channel is associated with positive operator-valued measurements, or POVMs (see, e.g., Ref.~\cite{NC10} for details). There is no quantum output in this case, so the channel is represented without the output wire. When $E_A$ is an element of a POVM on $A$, also known as effect, i.e., $0 \leq E \leq \id_A$, we represent it with the symbol
\begin{center}
	\begin{tikzpicture}
		\coordinate (start);
		\node[right=of start,shape = rounded rectangle, rounded rectangle left arc = none, minimum height = 20, draw] (state) {$E$};
		\draw [-] (start) to node [above] {$A$} (state) ;
	\end{tikzpicture}.
\end{center}

From these building blocks, one can construct more complex diagrams by connecting a symbol's output wire with another's input wire. For example, by connecting two channels in sequence, one obtains
\begin{equation}
	\begin{tikzpicture}[baseline=(channel.base)]
		\node[shape = rectangle, minimum height=20pt, draw] (channel) {$\mathcal{M}$};
		\node[right= of channel, shape = rectangle, minimum height=20pt, draw] (channelN) {$\mathcal{N}$};
		\coordinate[left=of channel] (A);
		\coordinate[right=of channelN] (C) {};
		\draw [-] (channel) to node [above] {$A$} (A) ;
		\draw [-] (channelN) to node [above] {$C$} (C) ;
		\draw [-] (channel) to node [above] {$B$} (channelN) ;
	\end{tikzpicture} = \mN_{B \to C} \circ \mM_{A \to B}\, .
\end{equation}
The special case when the first channel is a state is represented as
\begin{equation}
	\begin{tikzpicture}[baseline=(state.base)]
		\node[shape = rounded rectangle, rounded rectangle right arc = none, minimum height = 20, draw] (state) {$\rho$};
		\node[right= of state, shape = rectangle, minimum height=20pt, draw] (channel) {$\mathcal{M}$};
		\coordinate[right=of channel] (B) {};
		\draw [-] (state) to node [above] {$A$} (channel) ;
		\draw [-] (channel) to node [above] {$B$} (B) ;
	\end{tikzpicture} = \mM_{A \to B}(\rho_A)\, .
\end{equation}

On the other hand, an element of a POVM $E_A$ maps a quantum state $\rho_A$ into $\tr_A(E_A\rho_A)$, i.e., the probability of $E_A$ occurring on $\rho_A$, and this is the algebraic expression associated with the diagram below
\begin{equation}
	\begin{tikzpicture}[baseline=(state.base)]
		\node[shape = rounded rectangle, rounded rectangle right arc = none, minimum height = 20, draw] (state) {$\rho$};
		\node[right=of state, shape = rounded rectangle, rounded rectangle left arc = none, minimum height = 20, draw] (effect) {$E$};
		\draw [-] (state) to node [above] {$A$} (effect) ;
	\end{tikzpicture} = \tr_A(E_A \rho_A).
\end{equation}
More generally, an effect $E_A$ could act on only a part of a state $\rho_{AB}$: In this case, the diagram and the associated algebraic expression are
\begin{equation}\label{eq:par-eff}
	\begin{tikzpicture}[baseline = (state.base)]
		\node[shape = rounded rectangle, rounded rectangle right arc = none, minimum height=40pt, draw] (state) {$\rho$};
		\node[right=of $(state.north east)!0.2!(state.south east)$, shape = rounded rectangle, rounded rectangle left arc = none, minimum height = 20, draw] (effect) {$E$};
		\coordinate[right=of $(state.north east)!0.8!(state.south east)$]  (PA) {};
		\draw [-] ($(state.north east)!0.2!(state.south east)$) to node [above] {$A$} (effect) ;
		\draw [-] ($(state.north east)!0.8!(state.south east)$) to node [above] {$B$} (PA) ;
	\end{tikzpicture} = \tr_A[\rho_{AB}(E_A \otimes \id _B)].
\end{equation}
Therefore, every time an effect is composed with another diagram, the algebraic expression associated with it contains the trace over the common systems.

The expression above allows us to introduce one of the benefits of string diagrams. The diagram in Eq.~\eqref{eq:par-eff} shows only one output wire. Therefore, it can be simplified and written as the diagram of a (possibly unnormalized) state:
\begin{equation}
	\begin{tikzpicture}[baseline = (state.base)]
		\node[shape = rounded rectangle, rounded rectangle right arc = none, minimum height=40pt, draw] (state) {$\rho$};
		\node[shape = rounded rectangle, rounded rectangle right arc = none, minimum height=80pt, minimum width= 120pt, red, dashed, draw] at (state.east){};
		\node[right=of $(state.north east)!0.2!(state.south east)$, shape = rounded rectangle, rounded rectangle left arc = none, minimum height = 20, draw] (effect) {$E$};
		\coordinate[right=96pt of $(state.north east)!0.8!(state.south east)$]  (PA) {};
		\draw [-] ($(state.north east)!0.2!(state.south east)$) to node [above] {$A$} (effect) ;
		\draw [-] ($(state.north east)!0.8!(state.south east)$) to node [above] {$B$} (PA) ;
	\end{tikzpicture} = \begin{tikzpicture}[baseline = (state.base)]
		\node[shape = rounded rectangle, rounded rectangle right arc = none, minimum height = 40, draw] (state) {$\tr_A[\rho_{AB}(E_A \otimes \id _B)]$};
		\coordinate[right=of state] (end);
		\draw [-] (state) to node [above] {$B$} (end) ;
	\end{tikzpicture}.
\end{equation}

The tensor product of states, channels, and effects, also known as parallel composition, is represented by writing one symbol below the other, e.g.,
\begin{equation}
	\begin{tikzpicture}[baseline = (channel2.north)]
		\node[shape = rectangle, minimum height=20pt, draw] (channel) {$\mathcal{M}$};
		\coordinate[left=of channel] (A);
		\coordinate[right=of channel] (B) {};
		\draw [-] (channel) to node [above] {$A$} (A) ;
		\draw [-] (channel) to node [above] {$B$} (B) ;
		\node[below = 6pt of channel, shape = rectangle, minimum height=20pt, draw] (channel2) {$\mathcal{N}$};
		\coordinate[left=of channel2] (C);
		\coordinate[right=of channel2] (D) {};
		\draw [-] (channel2) to node [above] {$C$} (C) ;
		\draw [-] (channel2) to node [above] {$D$} (D) ;
	\end{tikzpicture} = \mM_{A \to B} \otimes \mN_{C \to D}.
\end{equation}

It is useful to denote two common channels with special diagrams. The identity channel $\mI_A$ on $A$, and the swap channel $\mS_{AB \to BA}$ from $AB$ to $BA$ are represented as
\begin{equation}
	\begin{tikzpicture}[baseline= (PA)]
		\coordinate (PA) {};
		\coordinate[right = of PA] (A);
		\draw [-] (PA) to node [above] {$A$} (A);
	\end{tikzpicture} = \mI_A \, , \quad
	\begin{tikzpicture}[baseline=($(PA')!0.5!(PB)$)]
		\node (PA') {$A$};
		\node [below= 12pt of PA'](PB) {$B$};
		\node[right =48pt of PA'] (endB) {$B$};
		\node[right = 48ptof PB](endA) {$A$};
		\draw [-] (PA') to [in = 180, out = 0, distance = 24 pt] (endA);
		\draw [-] (PB) to [in = 180, out = 0, distance = 24 pt] (endB);
	\end{tikzpicture} = \mS_{AB \to BA}.
\end{equation}
The choice of these symbols is intuitive, in that they convey that nothing is done on system $A$ in the former case, and that the systems $A$ and $B$ are swapped in the latter case.

Of particular interest is the symbol for the Choi state $\Phi_{A A'} = \sum_{x, y} \KBra{xx}{yy}_{AA'}$, where $\Set{\Ket{x}}_{A}$ is a fixed orthonormal basis for $A$ and $A'$ is a copy of $A$:
\begin{center}
	\begin{tikzpicture}
		\coordinate (PA) {};
		\coordinate[below=24 pt of PA] (PA');
		\coordinate[right = of PA] (A);
		\coordinate[right = of PA'](A');
		\draw [-] (PA) to [in = 180, out =180, distance = 24 pt] (PA');
		\draw [-] (PA) to node [above] {$A$} (A);
		\draw [-] (PA') to node [below] {$A'$} (A');
	\end{tikzpicture}\,.
\end{center}
It is particularly useful in proofs to differentiate $A$ and $A'$: This ensures that all the systems are properly accounted for. However, being $A'$ a copy of $A$, we will use either $\Phi_{A A'}$ or $\Phi_{A'A}$ depending on the need. In case of a bipartite system $AB$, the Choi state is $\sum_{a, b, a', b'}\KBra{abab}{a'b'a'b'}_{ABA'B'}$, where $\Set{\Ket{a}_A}$ and $\Set{\Ket{b}_B}$ are orthonormal basis for $A$ and $B$, respectively. This state can be rewritten as
\begin{equation}
	\begin{aligned}
		\Phi_{ABA'B'} & = \sum_{a, b,a'b'}\KBra{abab}{a'b'a'b'}_{ABA'B'}                                                                                                                \\
		              & = \sum_{a,b,a',b'} \KBra{a}{a'}_A \otimes \KBra{b}{b'}_B\otimes \KBra{a}{a'}_{A'}\otimes \KBra{b}{b'}_{B'}                                                      \\
		              & = \sum_{a,b, a',b'}(\mI_A\otimes  \mS_{A'B \to BA'}\otimes \mI_{B'})(\KBra{a}{a'}_A \otimes \KBra{a}{a'}_{A'}\otimes \KBra{b}{b'}_{B}\otimes \KBra{b}{b'}_{B'}) \\
		              & = (\mI_A\otimes  \mS_{A'B \to BA'}\otimes \mI_{B'})(\Phi_{AA'} \otimes \Phi_{BB'}).
	\end{aligned}
\end{equation}
This is easily represented diagrammatically as
\begin{equation}\label{eq:tchoi}
	\begin{tikzpicture}[baseline=($(A)!0.5!(A')$)]
		\coordinate (PA) {};
		\coordinate[below=24 pt of PA] (PA');
		\coordinate[right = of PA] (A);
		\coordinate[right = of PA'](A');
		\draw [-] (PA) to [in = 180, out =180, distance = 24 pt] (PA');
		\draw [-] (PA) to node [above] {$AB$} (A);
		\draw [-] (PA') to node [below] {$A'B'$} (A');
	\end{tikzpicture} = \begin{tikzpicture}[baseline=($(PA')!0.5!(PB)$)]
		\coordinate (PA) {};
		\coordinate[below=24 pt of PA] (PA');
		\coordinate[right = 48pt of PA] (A);
		\draw [-] (PA) to [in = 180, out =180, distance = 24 pt] (PA');
		\draw [-] (PA) to node [above] {$A$} (A);
		\coordinate [below= 24pt of PA'](PB);
		\coordinate[below=24 pt of PB] (PB');
		\coordinate[right = 48pt of PB'](B');
		\draw [-] (PB) to [in = 180, out =180, distance = 24 pt] (PB');
		\draw [-] (PB') to node [below] {$B'$} (B');
		\coordinate[right =48pt of PA'] (endB);
		\coordinate[right = 48ptof PB](endA);
		\draw [-] (PA') to [in = 180, out = 0, distance = 24 pt] node [below, near end] {$A$} (endA);
		\draw [-] (PB) to [in = 180, out = 0, distance = 24 pt] node [above, near end] {$B$} (endB);
	\end{tikzpicture}= \begin{tikzpicture}[baseline=($(A2)!0.5!(A3)$)]
		\coordinate (A1);
		\coordinate [below=24 pt of A1](A2);
		\coordinate [below=24 pt of A2](A3);
		\coordinate [below=24 pt of A3](A4);
		\coordinate [right=24 pt of A1](B1);
		\coordinate [right=24 pt of A2](B2);
		\coordinate [right=24 pt of A3](B3);
		\coordinate [right=24 pt of A4](B4);
		\draw[-] (A1) to node [above] {$A$} (B1);
		\draw[-] (A2) to node [above] {$B$} (B2);
		\draw[-] (A3) to node [above] {$A'$} (B3);
		\draw[-] (A4) to node [above] {$B'$} (B4);
		\draw [-] (A1) to [in = 180, out =180, distance = 48 pt] (A3);
		\draw [-] (A2) to [in = 180, out =180, distance = 48 pt] (A4);
	\end{tikzpicture} \,.
\end{equation}

To understand the choice of the symbol associated with the Choi state we need first to introduce the Choi effect $\Phi_{AA'}$, which has the same matrix as the Choi state and acts on a state $\rho_A$ as $\tr_A[\Phi_{AA'} (\rho_A \otimes \id_{A'})]$. The action on a state $\rho_{A'}$ is analogous. The symbol for the Choi effect is
\begin{center}
	\begin{tikzpicture}[baseline=($(PA)!0.5!(PA')$)]
		\coordinate (PA) {};
		\coordinate[below=24 pt of PA] (PA');
		\coordinate[right = of PA] (A);
		\coordinate[right = of PA'](A');
		\draw [-] (A) to [in = 0, out =0, distance = 24 pt] (A');
		\draw [-] (PA) to node [above] {$A$} (A);
		\draw [-] (PA') to node [below] {$A'$} (A');
	\end{tikzpicture}\,.
\end{center}

Two useful identities involving the Choi state and effect are (see, e.g., Refs~\cite{CK10, CHK14, CDKW15, CFS16, KU17, KU19, SSC21})
\begin{equation}\label{eq:snake}
	\begin{tikzpicture} [baseline= (PA')]
		\coordinate (PA) {};
		\coordinate[below=24 pt of PA] (PA');
		\coordinate[below=24 pt of PA'] (PA'');
		\coordinate[right = of PA] (A);
		\coordinate[right = of PA'](A');
		\coordinate[right = of PA''](A'');
		\draw [-] (A) to [in = 0, out =0, distance = 24 pt] (A');
		\draw [-] (PA') to [in = 180, out =180, distance = 24 pt] (PA'');
		\draw [-] (PA) to node [above] {$A$} (A);
		\draw [-] (PA') to node [above] {$A'$} (A');
		\draw [-] (PA'') to node [below] {$A$} (A'');
	\end{tikzpicture} = \begin{tikzpicture}[baseline= (PA)]
		\coordinate (PA) {};
		\coordinate[right = of PA] (A);
		\draw [-] (PA) to node [above] {$A$} (A);
	\end{tikzpicture} = \begin{tikzpicture}[baseline=(PA')]
		\coordinate (PA) {};
		\coordinate[below=24 pt of PA] (PA');
		\coordinate[below=24 pt of PA'] (PA'');
		\coordinate[right = of PA] (A);
		\coordinate[right = of PA'](A');
		\coordinate[right = of PA''](A'');
		\draw [-] (PA) to [in = 180, out =180, distance = 24 pt] (PA');
		\draw [-] (A') to [in = 0, out =0, distance = 24 pt] (A'');
		\draw [-] (PA) to node [above] {$A$} (A);
		\draw [-] (PA') to node [above] {$A'$} (A');
		\draw [-] (PA'') to node [below] {$A$} (A'');
	\end{tikzpicture}\, .
\end{equation}
These identities are called snake identities~\cite{CK17} and motivate the choice of the symbol for the Choi state and effect. Indeed, by connecting a Choi state to a Choi effect, one obtains a plain wire, i.e., the identity channel. To prove the snake identities, we use a test state $\rho_A$ and, to track all systems properly, we change the system in the output wire to $A''$, that is, we prove
\begin{equation}
	\begin{tikzpicture} [baseline= (PA')]
		\node[shape = rounded rectangle, rounded rectangle right arc = none, minimum height = 20, draw] (PA) {$\rho$};
		\coordinate[right = of PA] (A);
		\coordinate[below = of A](A');
		\coordinate[below = of A'](A'');
		\coordinate[left=24 pt of A'] (PA');
		\coordinate[left=24 pt of A''] (PA'');
		\draw [-] (A) to [in = 0, out =0, distance = 24 pt] (A');
		\draw [-] (PA') to [in = 180, out =180, distance = 24 pt] (PA'');
		\draw [-] (PA) to node [above] {$A$} (A);
		\draw [-] (PA') to node [above] {$A'$} (A');
		\draw [-] (PA'') to node [below] {$A''$} (A'');
	\end{tikzpicture} = \begin{tikzpicture}[baseline= (PA)]
		\node[shape = rounded rectangle, rounded rectangle right arc = none, minimum height = 20, draw] (PA) {$\rho$};
		\coordinate[right = of PA] (A);
		\draw [-] (PA) to node [above] {$A''$} (A);
	\end{tikzpicture} = \begin{tikzpicture}[baseline=(PA')]
		\node[shape = rounded rectangle, rounded rectangle right arc = none, minimum height = 20, draw] (PA) {$\rho$};
		\coordinate[right = of PA] (A);
		\coordinate[above = of A](A');
		\coordinate[above = of A'](A'');
		\coordinate[left=24 pt of A'] (PA');
		\coordinate[left=24 pt of A''] (PA'');
		\draw [-] (A) to [in = 0, out =0, distance = 24 pt] (A');
		\draw [-] (PA') to [in = 180, out =180, distance = 24 pt] (PA'');
		\draw [-] (PA) to node [below] {$A$} (A);
		\draw [-] (PA') to node [above] {$A'$} (A');
		\draw [-] (PA'') to node [above] {$A''$} (A'');
	\end{tikzpicture}\, .
\end{equation}

We start with the diagram on the left
\begin{equation}
	\begin{aligned}
		\tr_{AA'}[(\Phi_{AA'} \otimes \id_{A''})(\rho_A \otimes \Phi_{A'A''})] & =
		\sum_{x,y,a,b}\tr_{AA'}[(\KBra{xx}{yy}_{AA'} \otimes \id_{A''})(\rho_A \otimes \KBra{aa}{bb}_{A'A''})]                                                       \\
		                                                                       & = \sum_{x, y,a,b}\Braket{yy|(\rho \otimes \KBra{a}{b})|xx}_{AA'}  \KBra{a}{b}_{A''} \\
		                                                                       & = \sum_{x,y,a,b} \delta_{b, x}\delta_{a,y} \Braket{y|\rho|x}_{A}  \KBra{a}{b}_{A''} \\
		                                                                       & = \sum_{a,b} \Braket{a|\rho|b}_{A}  \KBra{a}{b}_{A''}                               \\
		                                                                       & = \rho_{A''}
	\end{aligned}
\end{equation}

The proof of the right identity is analogous. This is an example of the simplicity of the diagrammatic approach in comparison to the algebraic expression. In many of the upcoming proofs, instead of writing multiple lines of algebraic simplifications, we will simply `yank' a wire~\cite{Coe10, CP09}. It is important to notice again that this simplification does not reduce the rigor of a proof; it is equivalent to an algebraic proof.

Other useful diagrammatic identities involving the Choi state are the following. If we apply an element $E = \sum_{a, b} E_{a,b} \KBra{a}{b}_{A}$ of a POVM to one end of the Choi state, we obtain
\begin{equation}
	\begin{aligned}
		\tr_{A}[\Phi_{A A'} (E_A \otimes \id_{A'})] & =\sum_{x, y, a, b} E_{a, b} \tr_{A}[\KBra{xx}{yy}_{A A'} (\KBra{a}{b}_A \otimes \id_{A'})] \\
		                                            & =\sum_{x, y, a, b} E_{a, b} \Braket{y|(\KBra{a}{b})|x}_A \KBra{x}{y}_{A'}                  \\
		                                            & = \sum_{a, b}E_{a,b} \KBra{b}{a}_{A'} = E^T_{A'}.
	\end{aligned}
\end{equation}
Therefore, the Choi state maps an element $E$ of a POVM into the unnormalized state $E^T$. This is represented diagrammatically as
\begin{equation}\label{eq:t}
	\begin{tikzpicture}[baseline=($(start)!0.5!(end)$)]
		\coordinate (start);
		\node[right=of start,shape = rounded rectangle, rounded rectangle left arc = none, minimum height = 20, draw] (state) {$E$};
		\draw [-] (start) to node [above] {$A$} (state) ;
		\coordinate[below=24 pt of start] (end);
		\coordinate[right=80pt of end] (end2);
		\draw [-] (start) to [in = 180, out =180, distance = 24 pt] (end);
		\draw[-] (end) to node[below]{$A'$} (end2);
		\node[shape = rounded rectangle, rounded rectangle right arc = none, minimum height=80pt, minimum width= 120pt, red, dashed, draw] at ($(start)!0.5!(end)$){};
	\end{tikzpicture} =\begin{tikzpicture}[baseline=(state.base)]
		\node[shape = rounded rectangle, rounded rectangle right arc = none, minimum height = 20,  draw] (state) {$E^T$};
		\coordinate[right=of state] (end);
		\draw [-] (state) to node [above] {$A'$} (end) ;
	\end{tikzpicture} = \begin{tikzpicture}[baseline=($(start)!0.5!(end)$)]
		\coordinate (start);
		\node[right=of start,shape = rounded rectangle, rounded rectangle left arc = none, minimum height = 20, draw] (state) {$E$};
		\draw [-] (start) to node [below] {$A$} (state) ;
		\coordinate[above=24 pt of start] (end);
		\coordinate[right=80pt of end] (end2);
		\draw [-] (start) to [in = 180, out =180, distance = 24 pt] (end);
		\draw[-] (end) to node[above]{$A'$} (end2);
		\node[shape = rounded rectangle, rounded rectangle right arc = none, minimum height=80pt, minimum width= 120pt, red, dashed, draw] at ($(start)!0.5!(end)$){};
	\end{tikzpicture}\,.
\end{equation}
With a similar proof, one obtains
\begin{equation}\label{eq:t2}
	\begin{tikzpicture}[baseline=($(A)!0.5!(A1)$)]
		\node[shape = rounded rectangle, rounded rectangle right arc = none, minimum height = 20, draw] (state) {$\rho$};
		\coordinate [right= 24pt of state](A);
		\coordinate[below= 24pt of A] (A1);
		\draw[-] (state) to node [above] {$A$} (A);
		\draw [-] (A) to [in = 0, out =0, distance = 24 pt] (A1);
		\coordinate[left=80pt of A1] (A2);
		\draw[-] (A1) to node [below] {$A'$} (A2);
		\node[shape = rounded rectangle, rounded rectangle left arc = none, minimum height=80pt, minimum width= 120pt, red, dashed, draw] at ($(A)!0.5!(A1)$){};
	\end{tikzpicture} = 	\begin{tikzpicture}[baseline=(state.base)]
		\coordinate (start);
		\node[right=of start,shape = rounded rectangle, rounded rectangle left arc = none, minimum height = 20, draw] (state) {$\rho^T$};
		\draw [-] (start) to node [above] {$A'$} (state) ;
	\end{tikzpicture}= \begin{tikzpicture}[baseline=($(A)!0.5!(A1)$)]
		\node[shape = rounded rectangle, rounded rectangle right arc = none, minimum height = 20, draw] (state) {$\rho$};
		\coordinate [right= 24pt of state](A);
		\coordinate[above= 24pt of A] (A1);
		\draw[-] (state) to node [below] {$A$} (A);
		\draw [-] (A) to [in = 0, out =0, distance = 24 pt] (A1);
		\coordinate[left=80pt of A1] (A2);
		\draw[-] (A1) to node [above] {$A'$} (A2);
		\node[shape = rounded rectangle, rounded rectangle left arc = none, minimum height=80pt, minimum width= 120pt, red, dashed, draw] at ($(A)!0.5!(A1)$){};
	\end{tikzpicture}.
\end{equation}
\subsection{The Choi Isomorphism}
The Choi isomorphism maps quantum channels (or, more generally, linear operators that act on matrices) to bipartite matrices. Diagrammatically, there is an intuitive way to map a channel (one input and one output wire) into a bipartite matrix (two output wires): To transform an input wire in an output wire is enough to bend it in a way that its free end is directed towards the right, as depicted below.
\begin{equation}
	\begin{tikzpicture}[baseline=(channel.base)]
		\node[shape = rectangle, minimum height=20pt, draw] (channel) {$\mathcal{M}$};
		\coordinate[left=of channel] (A);
		\coordinate[right=of channel] (B) {};
		\draw [-] (channel) to node [above] {$A$} (A) ;
		\draw [-] (channel) to node [above] {$B$} (B) ;
	\end{tikzpicture} \quad \xrightarrow{\text{Choi}}\quad
	\begin{tikzpicture}[baseline = ($(A)!0.5!(A')$)]
		\node[shape = rectangle, minimum height=20pt, draw] (channel) {$\mathcal{M}$};
		\coordinate[left=of channel] (A);
		\coordinate[right=of channel] (B) {};
		\draw [-] (channel) to node [above] {$A$} (A) ;
		\draw [-] (channel) to node [above] {$B$} (B) ;

		\coordinate[below=24 pt of A] (A');
		\coordinate[below=24 pt of B] (PA');
		\draw [-] (A') to node [below] {$A'$} (PA');
		\draw [-] (A) to [in = 180, out =180, distance = 24 pt] (A');
		\node[shape = rounded rectangle, rounded rectangle right arc = none, minimum height=80pt, minimum width= 120pt, red, dashed, draw] at ($(A)!0.5!(A')$){};
	\end{tikzpicture}= \begin{tikzpicture}[baseline = (state.base)]
		\node[shape = rounded rectangle, rounded rectangle right arc = none, minimum height=40pt,  draw] (state) {$M$};
		\coordinate[right=of $(state.north east)!0.2!(state.south east)$] (endB) {};
		\coordinate[right=of $(state.north east)!0.8!(state.south east)$]  (PA) {};
		\draw [-] ($(state.north east)!0.2!(state.south east)$) to node [above] {$B$} (endB) ;
		\draw [-] ($(state.north east)!0.8!(state.south east)$) to node [above] {$A'$} (PA) ;
	\end{tikzpicture}.
\end{equation}
These diagrams coincide with the definition of Choi isomorphism that we gave in Section~\ref{sec:prelim}, i.e., the Choi matrix $M_{B \sep A'}$ associated with a channel $\mM_{A \to B}$ is defined as
\begin{equation}\label{eq:choi-def}
	M_{B\sep A'} = (\mM_{A \to B} \otimes \mI_{A'})(\Phi_{A A'}) \,.
\end{equation}
We have introduced the symbol `$\sep$' for the first time. With the definition of the Choi isomorphism we adopted, systems before `$\sep$' are output systems, while systems after `$\sep$' are input systems. As we will see in the next section, it is critical to properly keep track of input and output systems. Note that we decided to bend the input wire below the channel. Similarly, one could have chosen to bend the input wire above the channel, producing a different but still well defined Choi matrix $M_{A \sep B} = (\mI_{A} \otimes \mM_{A' \to B})(\Phi_{A A'})$. Throughout this article, we always use the convention that a Choi matrix is defined as in Eq.~\eqref{eq:choi-def}, i.e., by bending the wire below.

The Choi matrix of a quantum channel is Hermitian, positive semi-definite, and the marginal on the input system is the identity matrix, i.e., $\tr_{B} M_{B\sep A'} = \id_{A'}$~\cite{CdAP09}. Therefore, a Choi matrix is not a quantum state, as its trace is not 1. However, one can easily define a quantum state as the renormalized Choi matrix $\mu_{B \sep A'} = \frac{1}{d_A} M_{B \sep A'}$. Diagrammatically, one obtains
\begin{equation}
	\begin{tikzpicture}[baseline = (state.base)]
		\node[shape = rounded rectangle, rounded rectangle right arc = none, minimum height=40pt, draw] (state) {$\mu$};
		\coordinate[right=of $(state.north east)!0.2!(state.south east)$] (endB) {};
		\coordinate[right=of $(state.north east)!0.8!(state.south east)$]  (PA) {};
		\draw [-] ($(state.north east)!0.2!(state.south east)$) to node [above] {$B$} (endB) ;
		\draw [-] ($(state.north east)!0.8!(state.south east)$) to node [above] {$A'$} (PA) ;
	\end{tikzpicture} = \frac{1}{d_A} \begin{tikzpicture}[baseline = (state.base)]
		\node[shape = rounded rectangle, rounded rectangle right arc = none, minimum height=40pt,  draw] (state) {$M$};
		\coordinate[right=of $(state.north east)!0.2!(state.south east)$] (endB) {};
		\coordinate[right=of $(state.north east)!0.8!(state.south east)$]  (PA) {};
		\draw [-] ($(state.north east)!0.2!(state.south east)$) to node [above] {$B$} (endB) ;
		\draw [-] ($(state.north east)!0.8!(state.south east)$) to node [above] {$A'$} (PA) ;
	\end{tikzpicture} = \frac{1}{d_A} \begin{tikzpicture}[baseline = ($(A)!0.5!(A')$)]
		\node[shape = rectangle, minimum height=20pt, draw] (channel) {$\mathcal{M}$};
		\coordinate[left=of channel] (A);
		\coordinate[right=of channel] (B) {};
		\draw [-] (channel) to node [above] {$A$} (A) ;
		\draw [-] (channel) to node [above] {$B$} (B) ;
		\coordinate[below=24 pt of A] (A');
		\coordinate[below=24 pt of B] (PA');
		\draw [-] (A') to node [below] {$A'$} (PA');
		\draw [-] (A) to [in = 180, out =180, distance = 24 pt] (A');
	\end{tikzpicture} .
\end{equation}

One of the advantages of using string diagrams to approach the Choi isomorphism is that the definition of the inverse Choi isomorphism becomes straightforward. Indeed, one can undo the Choi isomorphism by bending the second wire back and using the snake identities in Eq.~\eqref{eq:snake}

\begin{equation}
	\begin{tikzpicture}[baseline = ($(A)!0.5!(A')$)]
		\node[shape = rectangle, minimum height=20pt, draw] (channel) {$\mathcal{M}$};
		\coordinate[left=of channel] (A);
		\coordinate[right=of channel] (B) {};
		\draw [-] (channel) to node [above] {$A$} (A) ;
		\draw [-] (channel) to node [above] {$B$} (B) ;

		\coordinate[below=24 pt of A] (A');
		\coordinate[below=24 pt of B] (PA');
		\draw [-] (A') to node [below] {$A'$} (PA');
		\draw [-] (A) to [in = 180, out =180, distance = 24 pt] (A');
		\node[shape = rounded rectangle, rounded rectangle right arc = none, minimum height=80pt, minimum width= 120pt, red, dashed, draw] at ($(A)!0.5!(A')$){};
	\end{tikzpicture} \quad \xrightarrow{\text{Choi inverse}}\begin{tikzpicture}[baseline = (A')]
		\node[shape = rectangle, minimum height=20pt, draw] (channel) {$\mathcal{M}$};
		\coordinate[left=of channel] (A);
		\coordinate[right=of channel] (B) {};
		\coordinate[right= 48pt of B] (endB);
		\draw [-] (channel) to node [above] {$A$} (A) ;
		\draw [-] (channel) to node [above] {$B$} (endB) ;
		\coordinate[below=24 pt of A] (A');
		\coordinate[below=24 pt of B] (PA');
		\draw [-] (A') to node[below] {$A'$}(PA');
		\draw [-] (A) to [in = 180, out =180, distance = 24 pt] (A');
		\coordinate[below=24 pt of A'] (AA);
		\coordinate[below=24 pt of PA'] (AA');
		\coordinate[left=48pt of AA] (endA);
		\draw [-] (endA) to node [below] {$A$} (AA');
		\draw [-] (PA') to [in = 0, out =0, distance = 24 pt] (AA');
		\coordinate (reference) at ($(A)!0.5!(A')$);
		\coordinate (reference2) at ($(channel)!0.5!(A)$);
		\node[shape = rounded rectangle, rounded rectangle right arc = none, minimum height=60pt, minimum width= 100pt, red, dashed, draw] at (reference2|-reference){};
	\end{tikzpicture} = \begin{tikzpicture}[baseline=(channel.base)]
		\node[shape = rectangle, minimum height=20pt, draw] (channel) {$\mathcal{M}$};
		\coordinate[left=of channel] (A);
		\coordinate[right=of channel] (B) {};
		\draw [-] (channel) to node [above] {$A$} (A) ;
		\draw [-] (channel) to node [above] {$B$} (B) ;
	\end{tikzpicture}.
\end{equation}

Therefore, the inverse Choi isomorphism transforms a Choi matrix $M_{B \sep A'}$ in the channel $\mM_{A \to B}$ defined as

\begin{equation}
	\begin{tikzpicture}[baseline = (state.base)]
		\node[shape = rounded rectangle, rounded rectangle right arc = none, minimum height=40pt,  draw] (state) {$M$};
		\coordinate[right=of $(state.north east)!0.2!(state.south east)$] (endB) {};
		\coordinate[right=of $(state.north east)!0.8!(state.south east)$]  (PA) {};
		\draw [-] ($(state.north east)!0.2!(state.south east)$) to node [above] {$B$} (endB) ;
		\draw [-] ($(state.north east)!0.8!(state.south east)$) to node [above] {$A'$} (PA) ;
	\end{tikzpicture}\xrightarrow{\text{Choi inverse}} \begin{tikzpicture}[baseline=(channel.base)]
		\node[shape = rectangle, minimum height=20pt, draw] (channel) {$\mathcal{M}$};
		\coordinate[left=of channel] (A);
		\coordinate[right=of channel] (B) {};
		\draw [-] (channel) to node [above] {$A$} (A) ;
		\draw [-] (channel) to node [above] {$B$} (B) ;
	\end{tikzpicture} = \begin{tikzpicture}[baseline = (state.base)]
		\node[shape = rounded rectangle, rounded rectangle right arc = none, minimum height=40pt,  draw] (state) {$M$};
		\coordinate[right=of $(state.north east)!0.2!(state.south east)$] (endB) {};
		\coordinate[right=of $(state.north east)!0.8!(state.south east)$]  (PA) {};
		\coordinate[below= 24pt of PA] (PAP) {};
		\coordinate[left = of PAP] (PAPstart){};
		\coordinate[right=24ptof endB] (endB2);
		\coordinate[left=24pt of PAPstart] (endA);
		\draw [-] ($(state.north east)!0.2!(state.south east)$) to node [above] {$B$} (endB2) ;
		\draw [-] ($(state.north east)!0.8!(state.south east)$) to node [above] {$A'$} (PA) ;
		\draw [-] (PA) to [in = 0, out = 0, distance = 24 pt] (PAP);
		\draw [-] (PAP) to node [below] {$A$} (endA);
	\end{tikzpicture}.
\end{equation}
With this definition, one quickly obtains the formula for the action of the channel $\mM_{A \to B}$ on a state $\rho_A$
\begin{equation}\label{eq:choi-ac}
	\mM_{A\to B}(\rho_A) = \begin{tikzpicture}[baseline=(state.base)]
		\node[shape = rounded rectangle, rounded rectangle right arc = none, minimum height = 20, draw] (state) {$\rho$};
		\node[right= of state, shape = rectangle, minimum height=20pt, draw] (channel) {$\mathcal{M}$};
		\coordinate[right=of channel] (B) {};
		\draw [-] (state) to node [above] {$A$} (channel) ;
		\draw [-] (channel) to node [above] {$B$} (B) ;
	\end{tikzpicture} = \begin{tikzpicture}[baseline = (PA)]
		\node[shape = rounded rectangle, rounded rectangle right arc = none, minimum height=40pt,  draw] (state) {$M$};
		\coordinate[right=of $(state.north east)!0.2!(state.south east)$] (endB) {};
		\coordinate[right=of $(state.north east)!0.8!(state.south east)$]  (PA) {};
		\coordinate[below= 24pt of PA] (PAP) {};
		\coordinate[left = of PAP] (PAPstart){};
		\coordinate[right=24ptof endB] (endB2);
		\node[shape = rounded rectangle, rounded rectangle right arc = none, minimum height = 20, draw] (endA) at (state |- PAPstart) {$\rho$};
		\draw [-] ($(state.north east)!0.2!(state.south east)$) to node [above] {$B$} (endB2) ;
		\draw [-] ($(state.north east)!0.8!(state.south east)$) to node [above] {$A'$} (PA) ;
		\draw [-] (PA) to [in = 0, out = 0, distance = 24 pt] (PAP);
		\draw [-] (PAP) to node [below] {$A$} (endA);
	\end{tikzpicture} = \begin{tikzpicture}[baseline = (state.base)]
		\node[shape = rounded rectangle, rounded rectangle right arc = none, minimum height=40pt,  draw] (state) {$M$};
		\coordinate[right=of $(state.north east)!0.2!(state.south east)$] (endB) {};
		\coordinate[right=of $(state.north east)!0.8!(state.south east)$]  (PA) {};
		\node[shape = rounded rectangle, rounded rectangle left arc = none, minimum height = 20, draw] (endA) at (PA) {$\rho^T$};
		\draw [-] ($(state.north east)!0.2!(state.south east)$) to node [above] {$B$} (endB) ;
		\draw [-] ($(state.north east)!0.8!(state.south east)$) to node [above] {$A'$} (endA) ;
	\end{tikzpicture},
\end{equation}
These diagrams correspond to the algebraic expressions for the inverse Choi map commonly used in literature:
\begin{equation}\label{eq:choi-ac2}
	\mM_{A \to B}(\rho_A) = \tr_{A'A}[(M_{B \sep A'} \otimes \rho_A)(\id_B \otimes \Phi_{A'A})] = \tr_{A'}[M_{B \sep A'} (\id_B \otimes\rho^T_{A'})].
\end{equation}

Thanks to the identities in Eq.~\eqref{eq:snake}, it is straightforward to see that the Choi map is indeed an isomorphism. In addition, diagrammatically, one immediately notices that the Choi matrix of a state $\rho_A$ is the state itself, since there is no input wire to bend, and the Choi matrix of an effect is the transpose of the matrix of the effect, see Eq.~\eqref{eq:t}.

\subsection{The Link Product}
Once the Choi isomorphism is well-defined, the natural step is translating operations between channels into operations between Choi matrices. Once again, we approach this problem with string diagrams. Then we show that the operation that naturally arises with this approach is the algebraic operation between Choi matrices known as link product.

Let us consider two channels $\mM_{A \to B}$ and $\mN_{B \to C}$, we call the channel resulting from their sequential composition $\mathcal{T}_{A \to C} = \mN_{B \to C} \circ \mM_{A \to B}$. Diagrammatically,
\begin{equation}\label{eq:seq-comp}
	\begin{tikzpicture}[baseline=(channel.base)]
		\node[shape = rectangle, minimum height=20pt, draw] (channel) {$\mathcal{T}$};
		\coordinate[left=of channel] (A);
		\coordinate[right=of channel] (B) {};
		\draw [-] (channel) to node [above] {$A$} (A) ;
		\draw [-] (channel) to node [above] {$C$} (B) ;
	\end{tikzpicture} = \begin{tikzpicture}[baseline=(channel.base)]
		\node[shape = rectangle, minimum height=20pt, draw] (channel) {$\mathcal{M}$};
		\coordinate[left=of channel] (A);
		\node[shape = rectangle, right=of channel, minimum height=20pt, draw] (channel2) {$\mathcal{N}$};
		\coordinate[right=of channel2] (B) {};
		\draw [-] (channel) to node [above] {$A$} (A) ;
		\draw [-] (channel) to node [above] {$B$} (channel2) ;
		\draw [-] (channel2) to node [above] {$C$} (B) ;
	\end{tikzpicture}.
\end{equation}
The Choi matrix $T_{C\sep A'}$ of $\mathcal{T}_{A \to C}$ is defined as
\begin{equation}
	\begin{tikzpicture}[baseline = ($(A)!0.5!(A')$)]
		\node[shape = rectangle, minimum height=20pt, draw] (channel) {$\mathcal{T}$};
		\coordinate[left=of channel] (A);
		\coordinate[right=of channel] (B) {};
		\draw [-] (channel) to node [above] {$A$} (A) ;
		\draw [-] (channel) to node [above] {$C$} (B) ;

		\coordinate[below=24 pt of A] (A');
		\coordinate[below=24 pt of B] (PA');
		\draw [-] (A') to node [below] {$A'$} (PA');
		\draw [-] (A) to [in = 180, out =180, distance = 24 pt] (A');
		\node[shape = rounded rectangle, rounded rectangle right arc = none, minimum height=80pt, minimum width= 120pt, red, dashed, draw] at ($(A)!0.5!(A')$){};
	\end{tikzpicture}= \begin{tikzpicture}[baseline = (state.base)]
		\node[shape = rounded rectangle, rounded rectangle right arc = none, minimum height=40pt,  draw] (state) {$T$};
		\coordinate[right=of $(state.north east)!0.2!(state.south east)$] (endB) {};
		\coordinate[right=of $(state.north east)!0.8!(state.south east)$]  (PA) {};
		\draw [-] ($(state.north east)!0.2!(state.south east)$) to node [above] {$C$} (endB) ;
		\draw [-] ($(state.north east)!0.8!(state.south east)$) to node [above] {$A'$} (PA) ;
	\end{tikzpicture}.
\end{equation}
From Eq.~\eqref{eq:seq-comp} we have that
\begin{equation}
	\begin{aligned}
		\begin{tikzpicture}[baseline = (state.base)]
			\node[shape = rounded rectangle, rounded rectangle right arc = none, minimum height=40pt,  draw] (state) {$T$};
			\coordinate[right=of $(state.north east)!0.2!(state.south east)$] (endB) {};
			\coordinate[right=of $(state.north east)!0.8!(state.south east)$]  (PA) {};
			\draw [-] ($(state.north east)!0.2!(state.south east)$) to node [above] {$C$} (endB) ;
			\draw [-] ($(state.north east)!0.8!(state.south east)$) to node [above] {$A'$} (PA) ;
		\end{tikzpicture} & = \begin{tikzpicture}[baseline=($(C)!0.5!(C')$)]
			                      \node[shape = rectangle, minimum height=20pt, draw] (channel) {$\mathcal{M}$};
			                      \node[right= of channel, shape = rectangle, minimum height=20pt, draw] (channelN) {$\mathcal{N}$};
			                      \coordinate[left=of channel] (A);
			                      \coordinate[right=of channelN] (C) {};
			                      \draw [-] (channel) to (A) ;
			                      \draw [-] (channelN) to node [above] {$C$} (C) ;
			                      \draw [-] (channel) to node [above] {$B$} (channelN) ;
			                      \coordinate[below=24 pt of A] (A');
			                      \coordinate[below=24 pt  of C] (C');
			                      \draw [-] (A') to node [above] {$A'$} (C') ;
			                      \draw [-] (A) to [in = 180, out = 180, distance = 24 pt] (A');
		                      \end{tikzpicture} \\ &= \begin{tikzpicture}[baseline=($(C)!0.5!(C')$)]
			\node[shape = rectangle, minimum height=20pt, draw] (channel) {$\mathcal{M}$};
			\node[right= of channel, shape = rectangle, minimum height=20pt, draw] (channelN) {$\mathcal{N}$};
			\coordinate[left=of channel] (A);
			\coordinate[right=of channelN] (C) {};
			\draw [-] (channel) to (A) ;
			\draw [-] (channelN) to node [above] {$C$} (C) ;
			\draw [-] (channel) to node [above] {$B$} (channelN) ;
			\coordinate[below=24 pt of A] (A');
			\coordinate[below=24 pt  of C] (C');
			\draw [-] (A') to node [above] {$A'$} (C') ;
			\draw [-] (A) to [in = 180, out = 180, distance = 24 pt] (A');
			\node[shape = rounded rectangle, rounded rectangle right arc = none, minimum height=80pt, minimum width= 120pt, red, dashed, draw] at ($(A)!0.5!(A')$){};
		\end{tikzpicture} \\ &= \begin{tikzpicture}[baseline = (state.base)]
			\node[shape = rounded rectangle, rounded rectangle right arc = none, minimum height=40pt,  draw] (state) {$M$};
			\node[shape = rectangle,right=of $(state.north east)!0.2!(state.south east)$, minimum height=20pt, draw] (channel) {$\mathcal{N}$};
			\coordinate[right=of channel] (endB) {};
			\coordinate[right=of $(state.north east)!0.8!(state.south east)$]  (PA) {};
			\coordinate (endA) at (endB|-PA);
			\draw [-] ($(state.north east)!0.2!(state.south east)$) to node [above] {$B$} (channel) ;
			\draw [-] ($(state.north east)!0.8!(state.south east)$) to node [above] {$A'$} (endA) ;
			\draw [-] (channel) to node [above] {$C$} (endB) ;
		\end{tikzpicture}.
	\end{aligned}
\end{equation}
So far, we have been able to express the Choi matrix of $\mathcal{T}_{A \to C}$ in terms of the Choi matrix of $\mM_{A \to B}$. To replace $\mN_{B \to C}$ with its Choi matrix, we insert one of the identities in Eq.~\eqref{eq:snake} between $M_{B \sep A'}$ and $\mN_{B \to C}$:
\begin{equation}\label{eq:link-p}
	\begin{aligned}
		\begin{tikzpicture}[baseline = (state.base)]
			\node[shape = rounded rectangle, rounded rectangle right arc = none, minimum height=40pt,  draw] (state) {$T$};
			\coordinate[right=of $(state.north east)!0.2!(state.south east)$] (endB) {};
			\coordinate[right=of $(state.north east)!0.8!(state.south east)$]  (PA) {};
			\draw [-] ($(state.north east)!0.2!(state.south east)$) to node [above] {$C$} (endB) ;
			\draw [-] ($(state.north east)!0.8!(state.south east)$) to node [above] {$A'$} (PA) ;
		\end{tikzpicture} & = \begin{tikzpicture}[baseline = (state.base)]
			                      \node[shape = rounded rectangle, rounded rectangle right arc = none, minimum height=40pt,  draw] (state) {$M$};
			                      \node[shape = rectangle,right=of $(state.north east)!0.2!(state.south east)$, minimum height=20pt, draw] (channel) {$\mathcal{N}$};
			                      \coordinate[right=of channel] (endB) {};
			                      \coordinate[right=of $(state.north east)!0.8!(state.south east)$]  (PA) {};
			                      \coordinate (endA) at (endB|-PA);
			                      \draw [-] ($(state.north east)!0.2!(state.south east)$) to node [above] {$B$} (channel) ;
			                      \draw [-] ($(state.north east)!0.8!(state.south east)$) to node [above] {$A'$} (endA) ;
			                      \draw [-] (channel) to node [above] {$C$} (endB) ;
		                      \end{tikzpicture} \\ &= \begin{tikzpicture}[baseline = ($(B1)!0.5!(B)$)]
			\node[shape = rounded rectangle, rounded rectangle right arc = none, minimum height=40pt,  draw] (state) {$M$};
			\coordinate[right=72 pt of $(state.north east)!0.2!(state.south east)$]  (B) {};
			\coordinate[right=72 pt of $(state.north east)!0.8!(state.south east)$]  (endA) {};
			\coordinate[above=24 pt of B] (B1);
			\coordinate[above=24 pt of B1] (B2);
			\coordinate [left=48 pt of B1](A1);
			\coordinate [left=48 pt of B2](A2);
			\draw[-] (A1) to node [above] {$B'$} (B1);
			\draw [-] (B1) to [in = 0, out = 0, distance = 24 pt] (B);
			\draw [-] (A2) to [in = 180, out = 180, distance = 24 pt] (A1);
			\node[shape = rectangle, minimum height=20pt, draw] (channel) at ($(A2)!0.5!(B2)$) {$\mathcal{N}$};
			\draw[-] (A2) to node [above] {$B$} (channel);
			\draw[-] (channel) to node [above] {$C$} (B2);
			\draw[-] ($(state.north east)!0.2!(state.south east)$) to node [above] {$B$} (B);
			\draw[-] ($(state.north east)!0.8!(state.south east)$) to node [above] {$A'$} (endA);
			\node[shape = rounded rectangle, rounded rectangle right arc = none, minimum height=50pt, minimum width= 75pt, red, dashed, draw] at ($(A2)!0.5!(A1)$){};
		\end{tikzpicture}\\ &= \begin{tikzpicture} [baseline=($(PB)!0.5!(PB')$)]
			\node[shape = rounded rectangle, rounded rectangle right arc = none, minimum height=40pt,  draw] (state) {$N$};
			\coordinate[right=48pt of $(state.north east)!0.2!(state.south east)$] (endC) {};
			\coordinate[right=of $(state.north east)!0.8!(state.south east)$]  (PB) {};
			\draw [-] ($(state.north east)!0.2!(state.south east)$) to node [above] {$C$} (endC) ;
			\draw [-] ($(state.north east)!0.8!(state.south east)$) to node [above] {$B'$} (PB) ;
			\coordinate [below= 24pt of PB] (PB');
			\draw [-] (PB) to [in = 0, out = 0, distance = 24 pt] (PB');
			\coordinate [below= 24pt of PB'] (refA);
			\coordinate (reference) at ($(PB')!0.5!(refA)$);
			\node[shape = rounded rectangle, rounded rectangle right arc = none, minimum height=40pt,  draw] (state2) at (state.center |- reference)  {$M$};
			\coordinate[right=48pt of $(state2.north east)!0.8!(state2.south east)$] (endA) {};
			\draw [-] ($(state2.north east)!0.2!(state2.south east)$) to node [above] {$B$} (PB') ;
			\draw [-] ($(state2.north east)!0.8!(state2.south east)$) to node [above] {$A'$} (endA) ;
		\end{tikzpicture}.
	\end{aligned}
\end{equation}
The algebraic expression associated with this diagram is
\begin{equation}
	T_{C \sep A'} = \tr_{B B'}[(N_{C \sep B'} \otimes M_{B \sep A'})(\id_C \otimes \Phi_{B' B} \otimes \id_{A'})].
\end{equation}
This operation is known as link product, and it is denoted with $N_{C \sep B'} * M_{B \sep A'}$~\cite{CdAP09}. The name `link' product is well-motivated from a diagrammatic point of view because this operation links the Choi matrices $N_{C \sep B'}$ and $M_{B \sep A'}$. Indeed, diagrammatically it is represented as a `bridge' that links two bipartite matrices, without any wire crossing, i.e., the second system of the first matrix is linked to the first system of the second matrix.

Eq.~\eqref{eq:link-p} shows that the link product of the Choi matrices of two channels is the Choi matrix of the sequential composition of these channels. Moreover, the link product of Choi matrices of channels always produces the Choi matrix of a channel, that is, a positive semi-definite matrix with the identity as marginal on the input state~\cite{CdAP09}.

As a final note, in Eq.~\eqref{eq:choi-ac} and Eq.~\eqref{eq:choi-ac2}, we expressed the action of a channel $\mM_{A \to B}$ on a state $\rho$ in terms of its Choi matrix $M_{B \sep A'}$. That expression is nothing but the link product $M_{B \sep A'} * \rho_A$, where $\rho_A$ is seen as the bipartite Choi matrix on $A$ and the trivial system, corresponding to the Hilbert space $\mathbb{C}$, associated with the channel that prepares $\rho_A$.
\begin{equation}
	\mM_{A\to B}(\rho_A) = \begin{tikzpicture}[baseline=(state.base)]
		\node[shape = rounded rectangle, rounded rectangle right arc = none, minimum height = 20, draw] (state) {$\rho$};
		\node[right= of state, shape = rectangle, minimum height=20pt, draw] (channel) {$\mathcal{M}$};
		\coordinate[right=of channel] (B) {};
		\draw [-] (state) to node [above] {$A$} (channel) ;
		\draw [-] (channel) to node [above] {$B$} (B) ;
	\end{tikzpicture} = \begin{tikzpicture}[baseline = (PA)]
		\node[shape = rounded rectangle, rounded rectangle right arc = none, minimum height=40pt,  draw] (state) {$M$};
		\coordinate[right=of $(state.north east)!0.2!(state.south east)$] (endB) {};
		\coordinate[right=of $(state.north east)!0.8!(state.south east)$]  (PA) {};
		\coordinate[below= 24pt of PA] (PAP) {};
		\coordinate[left = of PAP] (PAPstart){};
		\coordinate[right=24ptof endB] (endB2);
		\node[shape = rounded rectangle, rounded rectangle right arc = none, minimum height = 20, draw] (endA) at (state |- PAPstart) {$\rho$};
		\draw [-] ($(state.north east)!0.2!(state.south east)$) to node [above] {$B$} (endB2) ;
		\draw [-] ($(state.north east)!0.8!(state.south east)$) to node [above] {$A'$} (PA) ;
		\draw [-] (PA) to [in = 0, out = 0, distance = 24 pt] (PAP);
		\draw [-] (PAP) to node [below] {$A$} (endA);
	\end{tikzpicture} = M_{B \sep A'} * \rho_A.
\end{equation}

\subsection{The swap tensor product}
As seen in the previous subsection, sequential composition of channels is translated into the link product of Choi matrices. Parallel composition of channels is often overlooked or translated into the tensor product of Choi matrices~\cite{CdAP09, BCdAP16}. However, one needs to be careful with the order of the systems, as we show here. We start with two channels, $\mM_{A \to B}$ and $\mN_{C \to D}$, and analogously to what we did before, we define $\mathcal{T}_{AC \to BD} = \mM_{A \to B} \otimes \mN_{C \to D}$. Diagrammatically, we have
\begin{equation}
	\begin{tikzpicture}[baseline=($(A)!0.5!(C)$)]
		\node[shape = rectangle, minimum height=40pt,  draw] (channel) {$\mathcal{T}$};
		\coordinate[right=of $(channel.north east)!0.2!(channel.south east)$] (B) {};
		\coordinate[right=of $(channel.north east)!0.8!(channel.south east)$]  (D) {};
		\coordinate[left=of $(channel.north west)!0.2!(channel.south west)$] (A) {};
		\coordinate[left=of $(channel.north west)!0.8!(channel.south west)$]  (C) {};
		\draw[-] (A) to node[above] {$A$} ($(channel.north west)!0.2!(channel.south west)$);
		\draw[-] (C) to node[above] {$C$} ($(channel.north west)!0.8!(channel.south west)$);
		\draw[-] ($(channel.north east)!0.2!(channel.south east)$) to node[above] {$B$} (B);
		\draw[-] ($(channel.north east)!0.8!(channel.south east)$) to node[above] {$D$} (D);
	\end{tikzpicture} = \begin{tikzpicture}[baseline = ($(A)!0.5!(C)$)]
		\node[shape = rectangle, minimum height=20pt, draw] (channel) {$\mathcal{M}$};
		\coordinate[left=of channel] (A);
		\coordinate[right=of channel] (B) {};
		\draw [-] (channel) to node [above] {$A$} (A) ;
		\draw [-] (channel) to node [above] {$B$} (B) ;
		\node[below = 6pt of channel, shape = rectangle, minimum height=20pt, draw] (channel2) {$\mathcal{N}$};
		\coordinate[left=of channel2] (C);
		\coordinate[right=of channel2] (D) {};
		\draw [-] (channel2) to node [above] {$C$} (C) ;
		\draw [-] (channel2) to node [above] {$D$} (D) ;
	\end{tikzpicture}.
\end{equation}

As before, we compute the Choi matrix $T_{BD \sep A'C'}$ of $\mathcal{T}_{AC \to BD}$ with string diagrams:
\begin{equation}\label{eq:swap-prod}
	\begin{aligned}
		\begin{tikzpicture}[baseline = (state.base)]
			\node[shape = rounded rectangle, rounded rectangle right arc = none, minimum height=40pt,  draw] (state) {$T$};
			\coordinate[right=of $(state.north east)!0.2!(state.south east)$] (endB) {};
			\coordinate[right=of $(state.north east)!0.8!(state.south east)$]  (PA) {};
			\draw [-] ($(state.north east)!0.2!(state.south east)$) to node [above] {$BD$} (endB) ;
			\draw [-] ($(state.north east)!0.8!(state.south east)$) to node [above] {$A'C'$} (PA) ;
		\end{tikzpicture} & = \begin{tikzpicture}[baseline=($(C)!0.5!(A1)$)]
			                      \node[shape = rectangle, minimum height=40pt,  draw] (channel) {$\mathcal{T}$};
			                      \coordinate[right=of $(channel.north east)!0.2!(channel.south east)$] (B) {};
			                      \coordinate[right=of $(channel.north east)!0.8!(channel.south east)$]  (D) {};
			                      \coordinate[left=of $(channel.north west)!0.2!(channel.south west)$] (A) {};
			                      \coordinate[left=of $(channel.north west)!0.8!(channel.south west)$]  (C) {};
			                      \draw[-] (A) to node[above] {$A$} ($(channel.north west)!0.2!(channel.south west)$);
			                      \draw[-] (C) to node[above] {$C$} ($(channel.north west)!0.8!(channel.south west)$);
			                      \draw[-] ($(channel.north east)!0.2!(channel.south east)$) to node[above] {$B$} (B);
			                      \draw[-] ($(channel.north east)!0.8!(channel.south east)$) to node[above] {$D$} (D);
			                      \coordinate [below=24pt of C] (A1);
			                      \coordinate [below=24pt of D] (B1);
			                      \coordinate [below=24pt of A1] (A2);
			                      \coordinate [below=24pt of B1] (B2);
			                      \draw[-] (A1) to node[above] {$A'$}  (B1);
			                      \draw[-] (A2) to node[above] {$C'$}  (B2);
			                      \draw [-] (A) to [in = 180, out = 180, distance = 48 pt] (A1);
			                      \draw [-] (C) to [in = 180, out = 180, distance = 48 pt] (A2);
		                      \end{tikzpicture} \\&=\begin{tikzpicture}[baseline=($(A2)!0.5!(A3)$)]
			\coordinate (A1);
			\coordinate[below=24pt of A1] (A2);
			\coordinate[below=24pt of A2] (A3);
			\coordinate[below=24pt of A3] (A4);
			\coordinate[right=48pt of A1] (B1);
			\coordinate[below=24pt of B1] (B2);
			\coordinate[below=24pt of B2] (B3);
			\coordinate[below=24pt of B3] (B4);
			\node[shape = rectangle, minimum height=20pt, draw] (channelM) at ($(A1)!0.5!(B1)$) {$\mathcal{M}$};
			\node[shape = rectangle, minimum height=20pt, draw] (channelN) at ($(A2)!0.5!(B2)$) {$\mathcal{N}$};
			\draw[-] (A1) to node[above] {$A$}  (channelM);
			\draw[-] (channelM) to node[above] {$B$}  (B1);
			\draw[-] (A2) to node[above] {$C$}  (channelN);
			\draw[-] (channelN) to node[above] {$D$}  (B2);
			\draw[-] (A3) to node[above] {$A'$}  (B3);
			\draw[-] (A4) to node[above] {$C'$}  (B4);
			\draw [-] (A1) to [in = 180, out = 180, distance = 48 pt] (A3);
			\draw [-] (A2) to [in = 180, out = 180, distance = 48 pt] (A4);
		\end{tikzpicture}\\&=\begin{tikzpicture}[baseline=($(A2)!0.5!(A3)$)]
			\coordinate (A1);
			\coordinate[below=24pt of A1] (A2);
			\coordinate[below=24pt of A2] (A3);
			\coordinate[below=24pt of A3] (A4);
			\coordinate[right=48pt of A1] (B1);
			\coordinate[below=24pt of B1] (B2);
			\coordinate[below=24pt of B2] (B3);
			\coordinate[below=24pt of B3] (B4);
			\coordinate[right=48pt of B1] (C1);
			\coordinate[below=24pt of C1] (C2);
			\coordinate[below=24pt of C2] (C3);
			\coordinate[below=24pt of C3] (C4);
			\node[shape = rectangle, minimum height=20pt, draw] (channelM) at ($(A1)!0.5!(B1)$) {$\mathcal{M}$};
			\node[shape = rectangle, minimum height=20pt, draw] (channelN) at ($(A3)!0.5!(B3)$) {$\mathcal{N}$};
			\draw[-] (A1) to node[above] {$A$}  (channelM);
			\draw[-] (channelM) to node[above] {$B$}  (C1);
			\draw[-] (A2) to node[above] {$A'$}  (B2);
			\draw[-] (A3) to node[above] {$C$}  (channelN);
			\draw[-] (channelN) to node[above] {$D$}  (B3);
			\draw[-] (A4) to node[above] {$C'$}  (C4);
			\draw [-] (A1) to [in = 180, out = 180, distance = 24 pt] (A2);
			\draw [-] (B2) to [in = 180, out = 0, distance = 24 pt] node [below, near end] {$A'$}(C3);
			\draw [-] (B3) to [in = 180, out = 0, distance = 24 pt] node [above, near end] {$D$} (C2);
			\draw [-] (A3) to [in = 180, out = 180, distance = 24 pt] (A4);
			\node[shape = rounded rectangle, rounded rectangle right arc = none, minimum height=40pt, minimum width= 75pt, red, dashed, draw] at ($(A1)!0.3!(A2)$){};
			\node[shape = rounded rectangle, rounded rectangle right arc = none, minimum height=40pt, minimum width= 75pt, red, dashed, draw] at ($(A3)!0.3!(A4)$){};
		\end{tikzpicture}\\&=\begin{tikzpicture}[baseline=($(A2)!0.5!(A3)$)]
			\node[shape = rounded rectangle, rounded rectangle right arc = none, minimum height=40pt,  draw] (stateM) {$M$};
			\coordinate[right=48pt of $(stateM.north east)!0.2!(stateM.south east)$] (A1) {};
			\coordinate[right=48pt of $(stateM.north east)!0.8!(stateM.south east)$]  (A2) {};
			\coordinate [below=24pt of A2](A3);
			\coordinate [below=24pt of A3](A4);
			\coordinate (reference) at ($(A3)!0.5!(A4)$);
			\node[shape = rounded rectangle, rounded rectangle right arc = none, minimum height=40pt,  draw] (stateN) at (stateM|-reference) {$N$};
			\draw [-] ($(stateM.north east)!0.2!(stateM.south east)$) to node[above] {$B$}(A1);
			\draw [-] ($(stateM.north east)!0.8!(stateM.south east)$) to [in = 180, out = 0, distance = 24 pt] node [below, near end] {$A'$}(A3);
			\draw [-] ($(stateN.north east)!0.2!(stateN.south east)$) to [in = 180, out = 0, distance = 24 pt] node [above, near end] {$D$} (A2);
			\draw [-] ($(stateN.north east)!0.8!(stateN.south east)$) to node[below] {$C'$} (A4);
		\end{tikzpicture}.
	\end{aligned}
\end{equation}
The diagrammatic approach indicates that the tensor product of channels is \emph{not} associated with the tensor product of Choi matrices, but with the swap tensor product that we introduced in Appendix~\ref{sec:prod}:
\begin{equation}
	T_{BD \sep A'C'} = M_{B\sep A'} \boxtimes N_{D\sep C'} = (\mI_{B} \otimes \mS_{A'D \to DA'} \otimes \mI_{C'})(M_{B\sep A'} \otimes N_{D\sep C'}).
\end{equation}

As in the case of the link product, the Choi matrix of the tensor product of channels is the swap tensor product of the Choi matrices by construction, and every time we take the swap tensor product of two Choi matrices, of quantum channels we obtain the Choi matrix of a quantum channel, that is, a positive semi-definite matrix such that its marginal on the input systems is the identity matrix. This follows directly from the equalities in Eq.~\eqref{eq:swap-prod}.

\subsection{Choi-defined resource theories}
As a reminder, the CD operations associated with a set of free states are all and only the quantum channels such that their renormalized Choi matrix is a free state. CD operations are easily constructed from free states with the inverse Choi map, i.e., if $\mu_{B \sep A'}$ is a free state and a renormalized Choi matrix, which means $\tr_B \mu_{B\sep  A'} = \frac{1}{d_A}\id _{A'}$, then the channel $\mM_{A \to B}$ defined as
\begin{equation}\label{eq:cdo}
	\begin{tikzpicture}[baseline=(channel.base)]
		\node[shape = rectangle, minimum height=20pt, draw] (channel) {$\mathcal{M}$};
		\coordinate[left=of channel] (A);
		\coordinate[right=of channel] (B) {};
		\draw [-] (channel) to node [above] {$A$} (A) ;
		\draw [-] (channel) to node [above] {$B$} (B) ;
	\end{tikzpicture} = \begin{tikzpicture}[baseline = ($(state.north east)!0.8!(state.south east)$)]
		\node[shape = rounded rectangle, rounded rectangle right arc = none, minimum height=40pt,  draw] (state) {$\mu$};
		\node[left=2pt of state]  {$d_A$};
		\coordinate[right=72pt of $(state.north east)!0.2!(state.south east)$] (endB) {};
		\coordinate[below= 24pt of $(state.north east)!0.8!(state.south east)$] (PAP) {};
		\coordinate[left = 96pt of PAP] (PAPstart){};
		\draw [-] ($(state.north east)!0.2!(state.south east)$) to node [above] {$B$} (endB) ;
		\draw [-] ($(state.north east)!0.8!(state.south east)$) to [in = 0, out = 0, distance = 24 pt] (PAP);
		\draw [-] (PAP) to node [below] {$A$} (PAPstart);
		\node[shape = rectangle, minimum height=70pt, minimum width= 100pt, red, dashed, draw] at ($(state.north)!0.8!(state.south)$){};
	\end{tikzpicture}
\end{equation}
is a CD operation.

A quantum resource theory is a CDRT if its free operations coincide with the Choi-defined operations associated with its set of free states. This diagrammatically means that,
\begin{equation}\label{eq:cdtr}
	\begin{tikzpicture}[baseline= (channel.base)]
		\node[shape = rectangle, minimum height=20pt, draw] (channel) {$\mathcal{M}$};
		\coordinate[left=of channel] (A);
		\coordinate[right=of channel] (B) {};
		\draw [-] (channel) to node [above] {$A$} (A) ;
		\draw [-] (channel) to node [above] {$B$} (B) ;
	\end{tikzpicture} \quad \text{is free if and only if}\quad\begin{tikzpicture}[baseline = (state.base)]
		\node[shape = rounded rectangle, rounded rectangle right arc = none, minimum height=40pt, draw] (state) {$\mu$};
		\coordinate[right=of $(state.north east)!0.2!(state.south east)$] (endB) {};
		\coordinate[right=of $(state.north east)!0.8!(state.south east)$]  (PA) {};
		\draw [-] ($(state.north east)!0.2!(state.south east)$) to node [above] {$B$} (endB) ;
		\draw [-] ($(state.north east)!0.8!(state.south east)$) to node [above] {$A'$} (PA) ;
	\end{tikzpicture} = \frac{1}{d_A} \begin{tikzpicture}[baseline = ($(A)!0.5!(A')$)]
		\node[shape = rectangle, minimum height=20pt, draw] (channel) {$\mathcal{M}$};
		\coordinate[left=of channel] (A);
		\coordinate[right=of channel] (B) {};
		\draw [-] (channel) to node [above] {$A$} (A) ;
		\draw [-] (channel) to node [above] {$B$} (B) ;
		\coordinate[below=24 pt of A] (A');
		\coordinate[below=24 pt of B] (PA');
		\draw [-] (A') to node [below] {$A'$} (PA');
		\draw [-] (A) to [in = 180, out =180, distance = 24 pt] (A');
	\end{tikzpicture}\quad \text{is free.}
\end{equation}

We are now ready to prove the main result of this work diagrammatically.
\begin{theorem}\label{th:cdrt-diag}
	It is possible to construct a CDRT associated with a set of free states if and only if for all systems $A$, $B$
	\begin{enumerate}
		\item \label{enum:choi-diag} $\frac{1}{d_A}$ \begin{tikzpicture}[baseline = ($(PA)!0.5!(PA')$)]
			      \coordinate (PA) {};
			      \coordinate[below=24 pt of PA] (PA');
			      \coordinate[right = of PA] (A);
			      \coordinate[right = of PA'](A');
			      \draw [-] (PA) to [in = 180, out =180, distance = 24 pt] (PA');
			      \draw [-] (PA) to node [above] {$A$} (A);
			      \draw [-] (PA') to node [below] {$A'$} (A');
		      \end{tikzpicture} is free,
		\item \label{enum:seq-diag} if \begin{tikzpicture}[baseline=(state.base)]
			      \node[shape = rounded rectangle, rounded rectangle right arc = none, minimum height = 20, draw] (state) {$\rho$};
			      \coordinate[right=of state] (end);
			      \draw [-] (state) to node [above] {$A$} (end) ;
		      \end{tikzpicture} and \begin{tikzpicture}[baseline = (state.base)]
			      \node[shape = rounded rectangle, rounded rectangle right arc = none, minimum height=40pt, draw] (state) {$\mu$};
			      \coordinate[right=of $(state.north east)!0.2!(state.south east)$] (endB) {};
			      \coordinate[right=of $(state.north east)!0.8!(state.south east)$]  (PA) {};
			      \draw [-] ($(state.north east)!0.2!(state.south east)$) to node [above] {$B$} (endB) ;
			      \draw [-] ($(state.north east)!0.8!(state.south east)$) to node [above] {$A'$} (PA) ;
		      \end{tikzpicture} are free states, and \begin{tikzpicture}[baseline = ($(state.north east)!0.8!(state.south east)$)]
			      \node[shape = rounded rectangle, rounded rectangle right arc = none, minimum height=40pt,  draw] (state) {$\mu$};
			      \node[left=2pt of state]  {$d_A$};
			      \coordinate[right=72pt of $(state.north east)!0.2!(state.south east)$] (endB) {};
			      \coordinate[below= 24pt of $(state.north east)!0.8!(state.south east)$] (PAP) {};
			      \coordinate[left = 96pt of PAP] (PAPstart){};
			      \draw [-] ($(state.north east)!0.2!(state.south east)$) to node [above] {$B$} (endB) ;
			      \draw [-] ($(state.north east)!0.8!(state.south east)$) to [in = 0, out = 0, distance = 24 pt] (PAP);
			      \draw [-] (PAP) to node [below] {$A$} (PAPstart);
			      \node[shape = rectangle, minimum height=70pt, minimum width= 100pt, red, dashed, draw] at ($(state.north)!0.8!(state.south)$){};
		      \end{tikzpicture} is a quantum channel, then $d_A$\begin{tikzpicture}[baseline = (PA)]
			      \node[shape = rounded rectangle, rounded rectangle right arc = none, minimum height=40pt,  draw] (state) {$\mu$};
			      \coordinate[right=of $(state.north east)!0.2!(state.south east)$] (endB) {};
			      \coordinate[right=of $(state.north east)!0.8!(state.south east)$]  (PA) {};
			      \coordinate[below= 24pt of PA] (PAP) {};
			      \coordinate[left = of PAP] (PAPstart){};
			      \coordinate[right=24ptof endB] (endB2);
			      \node[shape = rounded rectangle, rounded rectangle right arc = none, minimum height = 20, draw] (endA) at (state |- PAPstart) {$\rho$};
			      \draw [-] ($(state.north east)!0.2!(state.south east)$) to node [above] {$B$} (endB2) ;
			      \draw [-] ($(state.north east)!0.8!(state.south east)$) to node [above] {$A'$} (PA) ;
			      \draw [-] (PA) to [in = 0, out = 0, distance = 24 pt] (PAP);
			      \draw [-] (PAP) to node [below] {$A$} (endA);
		      \end{tikzpicture} is a free state.
	\end{enumerate}
\end{theorem}
\begin{proof}
	We prove first the necessary conditions. We assume that the free states and the free operations form a well-defined Choi-defined resource theory.
	\begin{enumerate}
		\item \begin{tikzpicture}[baseline= (PA)]
			      \coordinate (PA) {};
			      \coordinate[right = of PA] (A);
			      \draw [-] (PA) to node [above] {$A$} (A);
		      \end{tikzpicture}  is a free operation in every resource theory. Therefore, in a CDRT, $\frac{1}{d_A}$ \begin{tikzpicture}[baseline = ($(PA)!0.5!(PA')$)]
			      \coordinate (PA) {};
			      \coordinate[below=24 pt of PA] (PA');
			      \coordinate[right = of PA] (A);
			      \coordinate[right = of PA'](A');
			      \draw [-] (PA) to [in = 180, out =180, distance = 24 pt] (PA');
			      \draw [-] (PA) to node [above] {$A$} (A);
			      \draw [-] (PA') to node [below] {$A'$} (A');
		      \end{tikzpicture} is free as well, see Eq.~\eqref{eq:cdtr}.
		\item If \begin{tikzpicture}[baseline=(state.base)]
			      \node[shape = rounded rectangle, rounded rectangle right arc = none, minimum height = 20, draw] (state) {$\rho$};
			      \coordinate[right=of state] (end);
			      \draw [-] (state) to node [above] {$A$} (end) ;
		      \end{tikzpicture} and \begin{tikzpicture}[baseline = (state.base)]
			      \node[shape = rounded rectangle, rounded rectangle right arc = none, minimum height=40pt, draw] (state) {$\mu$};
			      \coordinate[right=of $(state.north east)!0.2!(state.south east)$] (endB) {};
			      \coordinate[right=of $(state.north east)!0.8!(state.south east)$]  (PA) {};
			      \draw [-] ($(state.north east)!0.2!(state.south east)$) to node [above] {$B$} (endB) ;
			      \draw [-] ($(state.north east)!0.8!(state.south east)$) to node [above] {$A'$} (PA) ;
		      \end{tikzpicture} are free states, then \begin{tikzpicture}[baseline = ($(state.north east)!0.8!(state.south east)$)]
			      \node[shape = rounded rectangle, rounded rectangle right arc = none, minimum height=40pt,  draw] (state) {$\mu$};
			      \node[left=2pt of state]  {$d_A$};
			      \coordinate[right=72pt of $(state.north east)!0.2!(state.south east)$] (endB) {};
			      \coordinate[below= 24pt of $(state.north east)!0.8!(state.south east)$] (PAP) {};
			      \coordinate[left = 96pt of PAP] (PAPstart){};
			      \draw [-] ($(state.north east)!0.2!(state.south east)$) to node [above] {$B$} (endB) ;
			      \draw [-] ($(state.north east)!0.8!(state.south east)$) to [in = 0, out = 0, distance = 24 pt] (PAP);
			      \draw [-] (PAP) to node [below] {$A$} (PAPstart);
			      \node[shape = rectangle, minimum height=70pt, minimum width= 100pt, red, dashed, draw] at ($(state.north)!0.8!(state.south)$){};
		      \end{tikzpicture} is a free quantum channel (from Eq.~\eqref{eq:cdtr}), and \begin{tikzpicture}[baseline = ($(state.north east)!0.8!(state.south east)$)]
			      \node[shape = rounded rectangle, rounded rectangle right arc = none, minimum height=40pt,  draw] (state) {$\mu$};
			      \node[left=2pt of state]  {$d_A$};
			      \coordinate[right=72pt of $(state.north east)!0.2!(state.south east)$] (endB) {};
			      \coordinate[below= 24pt of $(state.north east)!0.8!(state.south east)$] (PAP) {};
			      \coordinate[left = 96pt of PAP] (PAPstart){};
			      \node[shape = rounded rectangle, rounded rectangle right arc = none, minimum height = 20, anchor=east, draw] (state2) at (PAPstart){$\rho$};
			      \draw [-] ($(state.north east)!0.2!(state.south east)$) to node [above] {$B$} (endB) ;
			      \draw [-] ($(state.north east)!0.8!(state.south east)$) to [in = 0, out = 0, distance = 24 pt] (PAP);
			      \draw [-] (PAP) to node [below] {$A$} (PAPstart);
			      \node[shape = rectangle, minimum height=70pt, minimum width= 100pt, red, dashed, draw] at ($(state.north)!0.8!(state.south)$){};
		      \end{tikzpicture} is a free state (it is the output of a free channel applied to a free state).
	\end{enumerate}
	We prove now that conditions~\ref{enum:choi-diag} and~\ref{enum:seq-diag} are sufficient to have a well-defined resource theory when the free operations are the Choi-defined operations. Recall that we assume that the set of free states under consideration is compatible with a minimal resource theory, i.e., closed under tensor product, partial tracing, and system swapping. Now, we demonstrate that the five conditions for a resource theory listed in Section~\ref{sec:prelim} are satisfied.
	\begin{enumerate}
		\item The identity channel is free. Indeed, since $\frac{1}{d_A}$ \begin{tikzpicture}[baseline = ($(PA)!0.5!(PA')$)]
			      \coordinate (PA) {};
			      \coordinate[below=24 pt of PA] (PA');
			      \coordinate[right = of PA] (A);
			      \coordinate[right = of PA'](A');
			      \draw [-] (PA) to [in = 180, out =180, distance = 24 pt] (PA');
			      \draw [-] (PA) to node [above] {$A$} (A);
			      \draw [-] (PA') to node [below] {$A'$} (A');
		      \end{tikzpicture} is free for all $A$, then \begin{tikzpicture}[baseline=(PA')]
			      \coordinate (PA) {};
			      \coordinate[below=24 pt of PA] (PA');
			      \coordinate[below=24 pt of PA'] (PA'');
			      \coordinate[right = of PA] (A);
			      \coordinate[right = of PA'](A');
			      \coordinate[right = of PA''](A'');
			      \draw [-] (PA) to [in = 180, out =180, distance = 24 pt] (PA');
			      \draw [-] (A') to [in = 0, out =0, distance = 24 pt] (A'');
			      \draw [-] (PA) to node [above] {$A$} (A);
			      \draw [-] (PA') to node [above] {$A'$} (A');
			      \draw [-] (PA'') to node [below] {$A$} (A'');
		      \end{tikzpicture} $=$ \begin{tikzpicture}[baseline= (PA)]
			      \coordinate (PA) {};
			      \coordinate[right = of PA] (A);
			      \draw [-] (PA) to node [above] {$A$} (A);
		      \end{tikzpicture} is a free operation (Eq.~\eqref{eq:cdo} and Eq.~\eqref{eq:snake}).
		\item The swap channel is free. Indeed, from condition~\ref{enum:choi-diag}, one has that $\frac{1}{d_Ad_B}$\begin{tikzpicture}[baseline=($(A2)!0.5!(A3)$)]
			      \coordinate (A1);
			      \coordinate [below=24 pt of A1](A2);
			      \coordinate [below=24 pt of A2](A3);
			      \coordinate [below=24 pt of A3](A4);
			      \coordinate [right=24 pt of A1](B1);
			      \coordinate [right=24 pt of A2](B2);
			      \coordinate [right=24 pt of A3](B3);
			      \coordinate [right=24 pt of A4](B4);
			      \draw[-] (A1) to node [above] {$A$} (B1);
			      \draw[-] (A2) to node [above] {$B$} (B2);
			      \draw[-] (A3) to node [above] {$A'$} (B3);
			      \draw[-] (A4) to node [above] {$B'$} (B4);
			      \draw [-] (A1) to [in = 180, out =180, distance = 48 pt] (A3);
			      \draw [-] (A2) to [in = 180, out =180, distance = 48 pt] (A4);
		      \end{tikzpicture} is a free state. Since the set of free states is closed under system swapping, one obtains that $\frac{1}{d_Ad_B}$\begin{tikzpicture}[baseline=($(A2)!0.5!(A3)$)]
			      \coordinate (A1);
			      \coordinate [below=24 pt of A1](A2);
			      \coordinate [below=24 pt of A2](A3);
			      \coordinate [below=24 pt of A3](A4);
			      \coordinate [right=48 pt of A1](B1);
			      \coordinate [right=48 pt of A2](B2);
			      \coordinate [right=48 pt of A3](B3);
			      \coordinate [right=48 pt of A4](B4);
			      \draw[-] (A1) to [in = 180, out = 0, distance = 24 pt] node [below, near end] {$A$} (B2);
			      \draw[-] (A2) to [in = 180, out = 0, distance = 24 pt] node [above, near end] {$B$} (B1);
			      \draw[-] (A3) to node [above] {$A'$} (B3);
			      \draw[-] (A4) to node [above] {$B'$} (B4);
			      \draw [-] (A1) to [in = 180, out =180, distance = 48 pt] (A3);
			      \draw [-] (A2) to [in = 180, out =180, distance = 48 pt] (A4);
		      \end{tikzpicture} is free too. From Eq.~\eqref{eq:cdo} and Eq.~\eqref{eq:snake}, it immediately follows that the CD operation associated with this state is \begin{tikzpicture}[baseline=($(PA')!0.5!(PB)$)]
			      \node (PA') {$A$};
			      \node [below= 12pt of PA'](PB) {$B$};
			      \node[right =48pt of PA'] (endB) {$B$};
			      \node[right = 48ptof PB](endA) {$A$};
			      \draw [-] (PA') to [in = 180, out = 0, distance = 24 pt] (endA);
			      \draw [-] (PB) to [in = 180, out = 0, distance = 24 pt] (endB);
		      \end{tikzpicture}.
		\item Discarding a system is free. Once again, $\frac{1}{d_A}$ \begin{tikzpicture}[baseline = ($(PA)!0.5!(PA')$)]
			      \coordinate (PA) {};
			      \coordinate[below=24 pt of PA] (PA');
			      \coordinate[right = of PA] (A);
			      \coordinate[right = of PA'](A');
			      \draw [-] (PA) to [in = 180, out =180, distance = 24 pt] (PA');
			      \draw [-] (PA) to node [above] {$A$} (A);
			      \draw [-] (PA') to node [below] {$A'$} (A');
		      \end{tikzpicture} is free for all $A$  and since the set of free operations is closed under partial tracing, then $\frac{1}{d_A}$ \begin{tikzpicture}[baseline = ($(PA)!0.5!(PA')$)]
			      \coordinate (PA) {};
			      \coordinate[below=24 pt of PA] (PA');
			      \coordinate[right = of PA] (A);
			      \coordinate[right = of PA'](A');
			      \draw [-] (PA) to [in = 180, out =180, distance = 24 pt] (PA');
			      \draw [-] (PA) to node [above] {$A$} (A);
			      \draw [-] (PA') to node [below] {$A'$} (A');
			      \node[shape = rounded rectangle, rounded rectangle left arc = none, minimum height = 20, anchor= west, draw] (endA) at (A) {$\id$};
		      \end{tikzpicture} is free as well. As above, from Eq.~\eqref{eq:cdo} and Eq.~\eqref{eq:snake}, one obtains that the CD operation associated with this state is \begin{tikzpicture}[baseline=(state.base)]
			      \coordinate (start);
			      \node[right=of start,shape = rounded rectangle, rounded rectangle left arc = none, minimum height = 20, draw] (state) {$\id$};
			      \draw [-] (start) to node [above] {$A$} (state) ;
		      \end{tikzpicture}, that is, the discarding channel.
		\item Sequential composition of free channels is free. Let \begin{tikzpicture}[baseline= (channel.base)]
			      \node[shape = rectangle, minimum height=20pt, draw] (channel) {$\mathcal{M}$};
			      \coordinate[left=of channel] (A);
			      \coordinate[right=of channel] (B) {};
			      \draw [-] (channel) to node [above] {$A$} (A) ;
			      \draw [-] (channel) to node [above] {$B$} (B) ;
		      \end{tikzpicture} and \begin{tikzpicture}[baseline= (channel.base)]
			      \node[shape = rectangle, minimum height=20pt, draw] (channel) {$\mathcal{N}$};
			      \coordinate[left=of channel] (A);
			      \coordinate[right=of channel] (B) {};
			      \draw [-] (channel) to node [above] {$B$} (A) ;
			      \draw [-] (channel) to node [above] {$C$} (B) ;
		      \end{tikzpicture} be free channels. Since the free operations are the CD operations, the states $\frac{1}{d_A}$\begin{tikzpicture}[baseline = ($(A)!0.5!(A')$)]
			      \node[shape = rectangle, minimum height=20pt, draw] (channel) {$\mathcal{M}$};
			      \coordinate[left=of channel] (A);
			      \coordinate[right=of channel] (B) {};
			      \draw [-] (channel) to node [above] {$A$} (A) ;
			      \draw [-] (channel) to node [above] {$B$} (B) ;
			      \coordinate[below=24 pt of A] (A');
			      \coordinate[below=24 pt of B] (PA');
			      \draw [-] (A') to node [below] {$A'$} (PA');
			      \draw [-] (A) to [in = 180, out =180, distance = 24 pt] (A');
		      \end{tikzpicture} and $\frac{1}{d_B}$\begin{tikzpicture}[baseline = ($(A)!0.5!(A')$)]
			      \node[shape = rectangle, minimum height=20pt, draw] (channel) {$\mathcal{N}$};
			      \coordinate[left=of channel] (A);
			      \coordinate[right=of channel] (B) {};
			      \draw [-] (channel) to node [above] {$B$} (A) ;
			      \draw [-] (channel) to node [above] {$C$} (B) ;
			      \coordinate[below=24 pt of A] (A');
			      \coordinate[below=24 pt of B] (PA');
			      \draw [-] (A') to node [below] {$B'$} (PA');
			      \draw [-] (A) to [in = 180, out =180, distance = 24 pt] (A');
		      \end{tikzpicture} are free. From condition~\ref{enum:choi-diag} and the closure of the set of free states under tensor product and system swapping, it follows that the state $\frac{1}{d_Ad_B}$\begin{tikzpicture}[baseline = ($(PA')!0.5!(C1)$)]
			      \node[shape = rectangle, minimum height=20pt, draw] (channel) {$\mathcal{N}$};
			      \coordinate[left=of channel] (A);
			      \coordinate[right=of channel] (B) {};
			      \draw [-] (channel) to node [above] {$B$} (A) ;
			      \draw [-] (channel) to node [above] {$C$} (B) ;
			      \coordinate[below=24 pt of A] (A');
			      \coordinate[below=24 pt of B] (PA');
			      \draw [-] (A) to [in = 180, out =180, distance = 24 pt] (A');
			      \coordinate[below=24pt of A'] (B1);
			      \coordinate[below=24pt of B1] (B2);
			      \coordinate[below=24pt of PA'] (C1);
			      \coordinate[below=24pt of C1] (C2);
			      \draw [-] (A') to [in = 180, out = 0, distance = 24 pt] node [above, near end] {$B'$} (C1);
			      \draw[-] (B1) to [in = 180, out = 0, distance = 24 pt] node [above, near end] {$A'$} (PA');
			      \draw[-] (B2) to node[below] {$A$} (C2);
			      \draw [-] (B1) to [in = 180, out =180, distance = 24 pt] (B2);
		      \end{tikzpicture} is free as well. Condition~\ref{enum:seq-diag} implies that the state
		      \begin{equation}\label{eq:th-seq-comp}
			      d_Ad_B \frac{1}{d_Ad_B} \frac{1}{d_A}\begin{tikzpicture}[baseline = ($(C1)!0.5!(C2)$)]
				      \node[shape = rectangle, minimum height=20pt, draw] (channel) {$\mathcal{N}$};
				      \coordinate[left=of channel] (A);
				      \coordinate[right=of channel] (B) {};
				      \draw [-] (channel) to node [above] {$B$} (A) ;
				      \draw [-] (channel) to node [above] {$C$} (B) ;
				      \coordinate[below=24 pt of A] (A');
				      \coordinate[below=24 pt of B] (PA');
				      \draw [-] (A) to [in = 180, out =180, distance = 24 pt] (A');
				      \coordinate[below=24pt of A'] (B1);
				      \coordinate[below=24pt of B1] (B2);
				      \coordinate[below=24pt of B2] (B3);
				      \coordinate[below=24pt of B3] (B4);
				      \coordinate[below=24pt of PA'] (C1);
				      \coordinate[below=24pt of C1] (C2);
				      \coordinate[below=24pt of C2] (C3);
				      \coordinate[below=24pt of C3] (C4);
				      \draw [-] (A') to [in = 180, out = 0, distance = 24 pt] node [above, near end] {$B'$} (C1);
				      \draw[-] (B1) to [in = 180, out = 0, distance = 24 pt] node [above, near end] {$A'$} (PA');
				      \draw[-] (B2) to node[above] {$A$} (C2);
				      \draw [-] (B1) to [in = 180, out =180, distance = 24 pt] (B2);
				      \draw[-] (B4) to node [below] {$A'$} (C4);
				      \node[shape = rectangle, minimum height=20pt, draw] at ($(B3)!0.5!(C3)$) (channel2) {$\mathcal{M}$};
				      \draw[-] (B3) to node [above] {$A$} (channel2);
				      \draw[-] (channel2) to node [above] {$B$} (C3);
				      \draw [-] (B3) to [in = 180, out =180, distance = 24 pt] (B4);
				      \draw [-] (C1) to [in = 0, out =0, distance = 48 pt] (C3);
				      \draw [-] (C2) to [in = 0, out =0, distance = 48 pt] (C4);
			      \end{tikzpicture} = \frac{1}{d_A}\begin{tikzpicture}[baseline=($(A)!0.5!(A')$)]
				      \node[shape = rectangle, minimum height=20pt, draw] (channel) {$\mathcal{M}$};
				      \node[right= of channel, shape = rectangle, minimum height=20pt, draw] (channelN) {$\mathcal{N}$};
				      \coordinate[left=of channel] (A);
				      \node[right=of channelN] (C) {$C$};
				      \draw [-] (channel) to (A) ;
				      \draw [-] (channelN) to (C) ;
				      \draw [-] (channel) to (channelN) ;
				      \coordinate[below=24 pt of A] (A');
				      \node (C') at (C|-A') {$A'$};
				      \draw [-] (A') to (C') ;
				      \draw [-] (A) to [in = 180, out = 180, distance = 24 pt] (A');
			      \end{tikzpicture}
		      \end{equation} is free. The equality above follows from multiple uses of Eq.~\eqref{eq:snake}. The CD operation associated with the state in Eq.~\eqref{eq:th-seq-comp} is \begin{tikzpicture}[baseline= (channel.base)]
			      \node[shape = rectangle, minimum height=20pt, draw] (channel) {$\mathcal{M}$};
			      \node[shape = rectangle, minimum height=20pt, right= of channel, draw] (channel2) {$\mathcal{N}$};
			      \coordinate[left=of channel] (A);
			      \coordinate[right=of channel2] (B) {};
			      \draw [-] (channel) to node [above] {$A$} (A) ;
			      \draw[-] (channel) to node [above] {$B$} (channel2);
			      \draw [-] (channel2) to node [above] {$C$} (B) ;
		      \end{tikzpicture}.
		\item Parallel composition of free channels is free. Let \begin{tikzpicture}[baseline= (channel.base)]
			      \node[shape = rectangle, minimum height=20pt, draw] (channel) {$\mathcal{M}$};
			      \coordinate[left=of channel] (A);
			      \coordinate[right=of channel] (B) {};
			      \draw [-] (channel) to node [above] {$A$} (A) ;
			      \draw [-] (channel) to node [above] {$B$} (B) ;
		      \end{tikzpicture} and \begin{tikzpicture}[baseline= (channel.base)]
			      \node[shape = rectangle, minimum height=20pt, draw] (channel) {$\mathcal{N}$};
			      \coordinate[left=of channel] (A);
			      \coordinate[right=of channel] (B) {};
			      \draw [-] (channel) to node [above] {$C$} (A) ;
			      \draw [-] (channel) to node [above] {$D$} (B) ;
		      \end{tikzpicture} be free channels. Since the free operations are the CD operations, the states $\frac{1}{d_A}$\begin{tikzpicture}[baseline = ($(A)!0.5!(A')$)]
			      \node[shape = rectangle, minimum height=20pt, draw] (channel) {$\mathcal{M}$};
			      \coordinate[left=of channel] (A);
			      \coordinate[right=of channel] (B) {};
			      \draw [-] (channel) to node [above] {$A$} (A) ;
			      \draw [-] (channel) to node [above] {$B$} (B) ;
			      \coordinate[below=24 pt of A] (A');
			      \coordinate[below=24 pt of B] (PA');
			      \draw [-] (A') to node [below] {$A'$} (PA');
			      \draw [-] (A) to [in = 180, out =180, distance = 24 pt] (A');
		      \end{tikzpicture} and $\frac{1}{d_C}$\begin{tikzpicture}[baseline = ($(A)!0.5!(A')$)]
			      \node[shape = rectangle, minimum height=20pt, draw] (channel) {$\mathcal{N}$};
			      \coordinate[left=of channel] (A);
			      \coordinate[right=of channel] (B) {};
			      \draw [-] (channel) to node [above] {$C$} (A) ;
			      \draw [-] (channel) to node [above] {$D$} (B) ;
			      \coordinate[below=24 pt of A] (A');
			      \coordinate[below=24 pt of B] (PA');
			      \draw [-] (A') to node [below] {$C'$} (PA');
			      \draw [-] (A) to [in = 180, out =180, distance = 24 pt] (A');
		      \end{tikzpicture} are free. Since the set of free states is closed under tensor product and system swapping, the state $\frac{1}{d_Ad_C}$ \begin{tikzpicture}[baseline = ($(A')!0.5!(B1)$)]
			      \node[shape = rectangle, minimum height=20pt, draw] (channel) {$\mathcal{M}$};
			      \coordinate[left=of channel] (A);
			      \coordinate[right=of channel] (B) {};
			      \draw [-] (channel) to node [above] {$A$} (A) ;
			      \draw [-] (channel) to node [above] {$B$} (B) ;
			      \coordinate[below=24 pt of A] (A');
			      \coordinate[below=24 pt of B] (PA');
			      \draw [-] (A) to [in = 180, out =180, distance = 24 pt] (A');
			      \coordinate[below=24 pt of A'] (B1);
			      \coordinate[below=24 pt of PA'] (C1);
			      \coordinate[below=24 pt of B1] (B2);
			      \coordinate[below=24 pt of C1] (C2);
			      \draw [-] (B2) to node [below] {$C'$} (C2);
			      \node[shape = rectangle, minimum height=20pt, draw] (channel2) at ($(B1)!0.5!(C1)$) {$\mathcal{N}$};
			      \coordinate (partial) at (channel2.east|-A');
			      \draw [-] (A') to (partial);
			      \draw [-] (channel2.east) to [in = 180, out =0, distance = 24 pt] node [above, near end] {$D$} (PA');
			      \draw [-] (partial) to [in = 180, out =0, distance = 24 pt] node [below, near end] {$A'$}(C1);
			      \draw [-] (B1) to (channel2);
			      \draw [-] (B1) to [in = 180, out =180, distance = 24 pt] (B2);
		      \end{tikzpicture} is free as well. Eq.~\eqref{eq:swap-prod} implies that the CD operation assocaited with this state is \begin{tikzpicture}[baseline = ($(A)!0.5!(C)$)]
			      \node[shape = rectangle, minimum height=20pt, draw] (channel) {$\mathcal{M}$};
			      \coordinate[left=of channel] (A);
			      \coordinate[right=of channel] (B) {};
			      \draw [-] (channel) to node [above] {$A$} (A) ;
			      \draw [-] (channel) to node [above] {$B$} (B) ;
			      \node[below = 6pt of channel, shape = rectangle, minimum height=20pt, draw] (channel2) {$\mathcal{N}$};
			      \coordinate[left=of channel2] (C);
			      \coordinate[right=of channel2] (D) {};
			      \draw [-] (channel2) to node [above] {$C$} (C) ;
			      \draw [-] (channel2) to node [above] {$D$} (D) ;
		      \end{tikzpicture}.
	\end{enumerate}
\end{proof}

\end{document}